\documentclass[final,11pt]{article}

\def\showauthornotes{0}
\def\showkeys{0}
\def\showdraftbox{0}

\usepackage{xspace,xcolor,enumerate}

\usepackage[full]{textcomp}

\usepackage{amsmath}

\usepackage{bm}

\usepackage{amsthm}

\newtheorem{theorem}{Theorem}[section]
\newtheorem*{theorem*}{Theorem}

\newtheorem{claim}[theorem]{Claim}
\newtheorem*{claim*}{Claim}

\newtheorem{proposition}[theorem]{Proposition}
\newtheorem*{proposition*}{Proposition}
\newtheorem{lemma}[theorem]{Lemma}
\newtheorem*{lemma*}{Lemma}
\newtheorem{corollary}[theorem]{Corollary}
\newtheorem{conjecture}[theorem]{Conjecture}
\newtheorem*{conjecture*}{Conjecture}

\newtheorem{fact}[theorem]{Fact}
\newtheorem*{fact*}{Fact}
\newtheorem{hypothesis}[theorem]{Hypothesis}
\newtheorem*{hypothesis*}{Hypothesis}

\theoremstyle{definition}
\newtheorem{definition}[theorem]{Definition}

\newtheorem{problem}[theorem]{Problem}

\newtheorem{remark}[theorem]{Remark}

\usepackage{fullpage}

\usepackage[varg]{txfonts} 
\renewcommand{\mathbb}{\varmathbb}

\ifnum\showkeys=1
\usepackage[color]{showkeys}
\fi

\usepackage[colorlinks,linkcolor=blue,filecolor=blue,citecolor=blue,urlcolor=blue]{hyperref}

\usepackage{prettyref}

\newcommand{\savehyperref}[2]{\texorpdfstring{\hyperref[#1]{#2}}{#2}}

\newrefformat{eq}{\savehyperref{#1}{\textup{(\ref*{#1})}}}
\newrefformat{lem}{\savehyperref{#1}{Lemma~\ref*{#1}}}
\newrefformat{def}{\savehyperref{#1}{Definition~\ref*{#1}}}
\newrefformat{thm}{\savehyperref{#1}{Theorem~\ref*{#1}}}
\newrefformat{cor}{\savehyperref{#1}{Corollary~\ref*{#1}}}
\newrefformat{cha}{\savehyperref{#1}{Chapter~\ref*{#1}}}
\newrefformat{sec}{\savehyperref{#1}{Section~\ref*{#1}}}
\newrefformat{app}{\savehyperref{#1}{Appendix~\ref*{#1}}}
\newrefformat{tab}{\savehyperref{#1}{Table~\ref*{#1}}}
\newrefformat{fig}{\savehyperref{#1}{Figure~\ref*{#1}}}
\newrefformat{hyp}{\savehyperref{#1}{Hypothesis~\ref*{#1}}}
\newrefformat{alg}{\savehyperref{#1}{Algorithm~\ref*{#1}}}
\newrefformat{rem}{\savehyperref{#1}{Remark~\ref*{#1}}}
\newrefformat{item}{\savehyperref{#1}{Item~\ref*{#1}}}
\newrefformat{step}{\savehyperref{#1}{step~\ref*{#1}}}
\newrefformat{conj}{\savehyperref{#1}{Conjecture~\ref*{#1}}}
\newrefformat{fact}{\savehyperref{#1}{Fact~\ref*{#1}}}
\newrefformat{prop}{\savehyperref{#1}{Proposition~\ref*{#1}}}
\newrefformat{claim}{\savehyperref{#1}{Claim~\ref*{#1}}}
\newrefformat{relax}{\savehyperref{#1}{Relaxation~\ref*{#1}}}
\newrefformat{red}{\savehyperref{#1}{Reduction~\ref*{#1}}}
\newrefformat{part}{\savehyperref{#1}{Part~\ref*{#1}}}

\newcommand{\Sref}[1]{\hyperref[#1]{\S\ref*{#1}}}

\usepackage{nicefrac}

\let\nfrac=\nicefrac
\let\ffrac=\flatfrac

\newcommand{\half}{\nicefrac12}

\usepackage{microtype}

\ifnum\showauthornotes=1
\newcommand{\Authornote}[2]{{\sffamily\small\color{red}{[#1: #2]}}}
\newcommand{\Authorcomment}[2]{{\sffamily\small\color{gray}{[#1: #2]}}}
\newcommand{\Authorstartcomment}[1]{\sffamily\small\color{gray}[#1: }

\newcommand{\Authorfnote}[2]{\footnote{\color{red}{#1: #2}}}
\newcommand{\Authorfixme}[1]{\Authornote{#1}{\textbf{??}}}
\newcommand{\Authormarginmark}[1]{\marginpar{\textcolor{red}{\fbox{\Large #1:!}}}}
\else
\newcommand{\Authornote}[2]{}
\newcommand{\Authorcomment}[2]{}
\newcommand{\Authorstartcomment}[1]{}

\newcommand{\Authorfnote}[2]{}
\newcommand{\Authorfixme}[1]{}
\newcommand{\Authormarginmark}[1]{}
\fi

\newcommand{\Dcomment}{\Authorcomment{D}}

\newcommand{\Pnote}{\Authornote{P}}

\newcommand{\Mnote}{\Authornote{M}}

\usepackage{boxedminipage}

\newenvironment{mybox}
{\center \noindent\begin{boxedminipage}{1.0\linewidth}}
{\end{boxedminipage}
\noindent
}

\newcommand{\paren}[1]{(#1)}
\newcommand{\Paren}[1]{\left(#1\right)}
\newcommand{\bigparen}[1]{\big(#1\big)}
\newcommand{\Bigparen}[1]{\Big(#1\Big)}

\newcommand{\Brac}[1]{\left[#1\right]}

\newcommand{\abs}[1]{\lvert#1\rvert}

\newcommand{\card}[1]{\lvert#1\rvert}

\newcommand{\set}[1]{\{#1\}}
\newcommand{\Set}[1]{\left\{#1\right\}}

\newcommand{\norm}[1]{\lVert#1\rVert}

\newcommand{\snorm}[1]{\norm{#1}^2}

\newcommand{\iprod}[1]{\langle#1\rangle}

\newcommand{\Esymb}{\mathbb{E}}
\newcommand{\Psymb}{\mathbb{P}}

\DeclareMathOperator*{\E}{\Esymb}

\DeclareMathOperator*{\ProbOp}{\Psymb}

\renewcommand{\Pr}{\ProbOp}

\newcommand{\Prob}[2][]{\Pr_{{#1}}\Set{#2}}

\newcommand{\Ex}[2][]{\E_{{#1}}\Brac{#2}}

\newcommand{\given}{\;\middle\vert\;}

\usepackage{dsfont}
\usepackage{yfonts}
\usepackage{mathrsfs}

\newcommand{\suchthat}{\;\middle\vert\;}

\newcommand{\tensor}{\otimes}

\newcommand{\textparen}[1]{\text{(#1)}}

\ifx\because\undefined
\newcommand{\because}[1]{\textparen{because #1}}
\else
\renewcommand{\because}[1]{\textparen{because #1}}
\fi

\newcommand{\bits}{\{0,1\}}

\newcommand{\super}[2]{#1^{\paren{#2}}}

\newcommand{\inv}[1]{{#1^{-1}}}

\newcommand{\vbig}{\vphantom{\bigoplus}}

\newcommand{\sm}{\setminus}

\newcommand{\defeq}{\stackrel{\mathrm{def}}=}     

\newcommand{\seteq}{\mathrel{\mathop:}=}

\newcommand{\from}{\colon}

\newcommand{\mper}{\,.}
\newcommand{\mcom}{\,,}

\newcommand\bdot\bullet

\ifx\mathds\undefined 
\newcommand{\Ind}{\mathbb I}
\else
\newcommand{\Ind}{\mathds 1}
\fi

\newcommand{\etal}{et al.\xspace}

\newcommand{\la}{\leftarrow}
\newcommand{\sse}{\subseteq}

\newcommand{\e}{\epsilon}
\newcommand{\eps}{\epsilon}

\DeclareMathOperator{\Inf}{Inf}

\DeclareMathOperator{\opt}{opt}
\DeclareMathOperator{\vol}{vol}
\DeclareMathOperator{\poly}{poly}

\newcommand{\N}{\mathbb N}
\newcommand{\R}{\mathbb R}

\newcommand{\problemmacro}[1]{\texorpdfstring{\textsc{#1}}{#1}\xspace}

\newcommand{\uniquegames}{\problemmacro{Unique Games}}
\newcommand{\maxcut}{\problemmacro{Max Cut}}

\newcommand{\vertexcover}{\problemmacro{Vertex Cover}}
\newcommand{\balancedseparator}{\problemmacro{Balanced Separator}}

\newcommand{\sparsestcut}{\problemmacro{Sparsest Cut}}

\newcommand{\smallsetexpansion}{\problemmacro{Small-Set Expansion}}
\newcommand{\minbisection}{\problemmacro{Min Bisection}}

\newcommand{\cA}{\mathcal A}

\newcommand{\cE}{\mathcal E}

\newcommand{\cG}{\mathcal G} 
 
\newcommand{\cI}{\mathcal I}

\newcommand{\cP}{\mathcal P}

\newcommand{\cV}{\mathcal V}

\newcommand{\cZ}{\mathcal Z}

\renewcommand{\leq}{\leqslant}
\renewcommand{\le}{\leqslant}
\renewcommand{\geq}{\geqslant}
\renewcommand{\ge}{\geqslant}

\ifnum\showdraftbox=1
\newcommand{\draftbox}{\begin{center}
  \fbox{%
    \begin{minipage}{2in}%
      \begin{center}%
          \Large\textsc{Working Draft}\\%
        Please do not distribute%
      \end{center}%
    \end{minipage}%
  }%
\end{center}
\vspace{0.2cm}}
\else
\newcommand{\draftbox}{}
\fi

\let\epsilon=\varepsilon

\numberwithin{equation}{section}

\let\origparagraph\paragraph
\renewcommand{\paragraph}[1]{\origparagraph{#1.}}

\newcommand{\DSstore}[2]{%
  \global\expandafter \def \csname DSMEMORY #1 \endcsname{#2}%
}

\newcommand{\DSload}[1]{%
  \csname DSMEMORY #1 \endcsname%
}

\newcommand{\DSnewlabel}[1]{%
  \newcommand\DScurrentlabel{#1}%
  \DSoldlabel{#1}%
}

\newcommand{\DSdummylabel}[1]{}

\newcommand{\torestate}[1]{%
  \let\DSoldlabel\label%
  \let\label\DSnewlabel%
  #1%
  \DSstore{\DScurrentlabel}{#1}%
  \let\label\DSoldlabel%
}

\newcommand{\restatetheorem}[1]{%
  \let\DSoldlabel\label
  \let\label\DSdummylabel
  \begin{theorem*}[Restatement of \prettyref{#1}]
    \DSload{#1}
  \end{theorem*}
  \let\label\DSoldlabel
}

\newcommand{\restatelemma}[1]{%
  \let\DSoldlabel\label
  \let\label\DSdummylabel
  \begin{lemma*}[Restatement of \prettyref{#1}]
    \DSload{#1}
  \end{lemma*}
  \let\label\DSoldlabel
}

\newcommand{\restateprop}[1]{%
  \let\DSoldlabel\label
  \let\label\DSdummylabel
  \begin{proposition*}[Restatement of \prettyref{#1}]
    \DSload{#1}
  \end{proposition*}
  \let\label\DSoldlabel
}

\newcommand{\restatefact}[1]{%
  \let\DSoldlabel\label
  \let\label\DSdummylabel
  \begin{fact*}[Restatement of \prettyref{#1}]
    \DSload{#1}
  \end{fact*}
  \let\label\DSoldlabel
}

\newcommand{\restate}[1]{%
  \let\DSoldlabel\label
  \let\label\DSdummylabel
  \DSload{#1}
  \let\label\DSoldlabel
}

\let\pref=\prettyref
\let\mla=\minimumlineararrangement
\let\BalancedSeparator=\balancedseparator

\newcommand{\inparen}[1]{\left(#1\right)}
\newcommand{\inbrace}[1]{\left\{#1\right\}}
\newcommand{\insquare}[1]{\left[#1\right]}

\newcommand{\rep}[1]{\insquare{#1}}

\newcommand{\betabias}{\set{\bot,\top}_\beta}
\newcommand{\B}{\{0,1\}}

\renewcommand{\vol}{\mu}

\newcommand{\glorifiedmarkov}{Glorified Markov Inequality\xspace}

\DeclareMathOperator{\surf}{\partial}
\DeclareMathOperator{\cond}{\Phi}  %

\newcommand{\tOmega}{\tilde\Omega}
\newcommand{\tC}{\tilde C}
\newcommand{\tA}{\tilde A}
\newcommand{\tB}{\tilde B}
\newcommand{\tU}{\tilde U}

\newcommand{\tv}{\tilde v}

\newcommand{\uginst}{\UG}

\newcommand{\SSE}{\problemmacro{Small-Set Expansion}\xspace}

\newcommand{\ssehard}{\textsf{SSE}-hard\xspace}
\newcommand{\nphard}{\textsf{NP}-hard\xspace}

\newcommand{\NP}{\textsf{NP}\xspace}
\newcommand{\UGCexpand}{Unique Games Conjecture\xspace}

\newcommand{\ugcomp}{\eps}
\newcommand{\ugsound}{\eta}

\newcommand{\yes}{\textsc{Yes}\xspace}

\newcommand{\no}{\textsc{No}\xspace}

\newcommand{\case}{{\sf Case}}
\newcommand{\UG}{\mathcal U}

\newcommand{\SSEH}{Small-Set Expansion Hypothesis\xspace}

\title{%
  Reductions Between Expansion Problems
}%

\author{%
  Prasad Raghavendra\thanks{%
Georgia Institute of Technology, Atlanta, GA.  Research done while visiting Princeton University}%
  \and David Steurer\thanks{%
    Microsoft Research New England, Cambridge, MA.  Research done at Princeton University supported by NSF Grants CCF-0832797, 0830673, and 0528414.}%
  \and Madhur Tulsiani\thanks{%
    Princeton University and Institute for Advanced Study, Princeton, NJ. Work supported by NSF grant CCF-0832797
and IAS Sub-contract no. 00001583.
  }%
}

\date{\today}

\begin{document}

\sloppy

\maketitle

\draftbox
\setcounter{page}{1}
\thispagestyle{empty}
\begin{abstract}

The \emph{Small-Set Expansion Hypothesis} (Raghavendra, Steurer, STOC 2010)
is a natural hardness assumption concerning the problem of approximating
the edge expansion of small sets 
in graphs.
This hardness assumption is closely connected to the \emph{Unique Games
  Conjecture} (Khot, STOC 2002).
In particular,
the Small-Set Expansion Hypothesis implies the Unique Games Conjecture
(Raghavendra, Steurer, STOC 2010).

Our main result is that the Small-Set Expansion Hypothesis
is in fact equivalent to a variant of the Unique Games Conjecture.
More precisely, the hypothesis is equivalent to the Unique Games Conjecture restricted
to instance with a fairly mild condition on the expansion of small sets.
Alongside, we obtain the first strong hardness of approximation results for
the \textsc{Balanced Separator} and \textsc{Minimum Linear Arrangement}
problems.
Before, no such hardness was known for these problems even assuming the
Unique Games Conjecture.  

These results not only establish the Small-Set Expansion Hypothesis as a natural unifying hypothesis
that implies the Unique Games Conjecture, all its consequences and, in addition,
hardness results for other problems like \textsc{Balanced Separator} and
\textsc{Minimum Linear Arrangement}, but our results also show that the
Small-Set Expansion Hypothesis problem lies at the combinatorial heart of the Unique Games Conjecture.

The key technical ingredient is a new way of exploiting the structure of
the  \textsc{Unique Games} instances  obtained from the \SSEH via
(Raghavendra, Steurer, 2010).  
This additional structure allows us to modify standard reductions
in a way that essentially destroys their local-gadget nature.
Using this modification, we can 
argue about the expansion in the graphs produced by the reduction without
relying on expansion properties of the underlying \uniquegames instance (which
would be impossible for a local-gadget reduction).
\end{abstract}

\newpage

\tableofcontents
\thispagestyle{empty}

\clearpage
\setcounter{page}{1}

\section{Introduction}

Finding small vertex or edge separators in a graph is a fundamental
computational task.  Even from a purely theoretical standpoint, the
phenomenon of vertex and edge expansion -- the lack of good vertex and
edge separators, has had numerous implications in all branches of
theoretical computer science.  Yet, the computational complexity of detecting and
approximating expansion, or finding good vertex and edge separators in graphs is not very well understood.  

Among the two notions of expansion, this work will concern mostly with
edge expansion.  For simplicity, let us first consider the case of a
$d$-regular graph $G = (V,E)$.
The edge expansion of a subset of vertices $S \sse V$ measures the
fraction of edges that leave $S$.  Formally, the edge expansion $\cond(S)$
of a (non-empty) subset $S \sse V$ is defined as,
$$ \cond_G(S) = \frac{|E(S, V\sm S)|}{d|S|} \mcom$$
where $E(S,V \sm S)$ denotes the set of edges with one endpoint in $S$ and
the other endpoint in $V \sm S$.  The conductance or the Cheeger's constant associated with the graph
$G$ is the minimum of $\cond(S)$ over all sets $S$ with at most half the
vertices, i.e.,
$$ \cond_G = \min_{|S| \leq \nfrac{n}{2}} \cond_G(S) \mper $$
These notions of conductance can be extended naturally to non-regular graphs, and finally to
arbitrary weighted graphs (see \pref{sec:preliminaries}).  Henceforth, in this section, for a subset of vertices $S$ in a graph $G$
we will use the notation $\mu(S)$ to denote the normalized set size, i.e.,
$\mu(S) = |S|/n$ in a $n$ vertex graph.

The problem of approximating the quantity $\cond_G$ for a graph $G$, also referred to as the
the uniform \sparsestcut (equivalent within a factor of $2$), is among the fundamental problems in
approximation algorithms.  Efforts towards approximating $\cond_G$
have led to a rich body of work with strong connections to spectral
techniques and metric embeddings.

The first approximation for conductance was obtained by discrete
analogues of the Cheeger inequality \cite{Cheeger70} shown
by Alon-Milman \cite{AlonM85} and Alon \cite{Alon86}.  
Specifically, Cheeger's inequality relates the conductance $\cond_G$ to the second eigenvalue of the adjacency matrix of the graph -- an efficiently computable quantity.  This yields an approximation algorithm for $\cond_G$, one that is used heavily in practice for graph partitioning.
However, the approximation for $\cond_G$ obtained via Cheeger's
inequality is poor in terms
of a approximation ratio, especially when the value of $\cond_G$ is small.
An $O(\log n)$ approximation algorithm for $\cond_G$ was obtained by
Leighton and Rao \cite{LeightonR99}.  Later work by Linial \etal
\cite{LinialLR95} and Aumann and Rabani \cite{AumannR98} established a strong connection between the
\sparsestcut problem and the theory of metric spaces, in turn
spurring a large and rich body of literature.  More
recently, in a breakthrough result Arora \etal \cite{AroraRV04}
obtained an $O(\sqrt{\log n})$ approximation for the problem using
semidefinite programming techniques.

\paragraph{Small Set Expansion} It is easy to see that $\cond_G$ is a fairly coarse measure of edge expansion,
in that it is the worst case edge expansion over sets $S$ of all
sizes.  In a typical graph (say a random $d$-regular graph), smaller
sets of vertices expand to a larger extent than sets with half the
vertices.  For instance, all sets $S$ of $n/1000$ vertices in a random
$d$-regular graph have $\cond(S) \geq 0.99$ with very high
probability, while the conductance~$\cond_G$ of the entire
graph is roughly $\half$.

A more refined measure of the edge expansion of a graph is its expansion
profile.  Specifically, for a graph $G$ the expansion profile is given
by the curve 
$$ \cond_{G}(\delta) = \min_{\mu(S) \le \delta} \cond(S) \qquad \qquad
\forall \delta \in [0,\nfrac{1}{2}] \mper$$
The problem of approximating the expansion profile has received much
less attention, and is seemingly far less tractable.  The second eigenvalue $\lambda_2$ fails to approximate the expansion of small sets in
graphs.  On one hand, even with the largest possible spectral gap, the
Cheeger's inequality cannot yield a lower bound greater than
$\nfrac{1}{2}$ for the conductance $\cond_G(\delta)$.
More importantly, there exists graphs such as hypercube where
$\cond_G$ is small (say $\eps$),  yet every sufficiently small set has
near perfect expansion ($\cond(S) \geq
1-\eps$).  This implies that $\cond_G$ (and the second eigenvalue
$\lambda_2$) does not yield any information about expansion of small
sets.

In a recent work, Raghavendra, Steurer, and Tetali~\cite{RaghavendraST10}
give a polynomial-time algorithm based on semidefinite programming for this
problem.
Roughly speaking, the approximation guarantee of their algorithm for
$\cond_G(\delta)$ is similar to the one given by Cheeger's inequality for
$\cond_G$, except with the approximation degrading by a
$\log\nfrac{1}{\delta}$ factor.  In particular, the approximation gets
worse as the size of the sets considered gets smaller.

In the regime when $\Phi_G(\delta)$ tends to zero as a function of the
instance size $n$, an $O(\log n)$-approximation follows from the framework of
R\"acke~\cite{Raecke08}. 
Very recently, this approximation has been improved to a $O(\sqrt{\log
  n\cdot \log (\nfrac 1\delta)})$-approximation \cite{BansalKMNNS10}.
Our work focuses on the regime when $\Phi_G(\delta)$ is not a function of
the instance size $n$.
In this regime, the algorithm of \cite{RaghavendraST10} gives the best
known approximation for the expansion profile $\Phi_g(\delta)$.

In summary, the current state-of-the-art algorithms for approximating the
expansion profile of a graph are still very far from satisfactory.
Specifically, the following hypothesis is consistent with the known
algorithms for approximating expansion profile.
\begin{hypothesis*}[\SSEH, \cite{RaghavendraS10}]
  For every constant $\eta > 0$, there exists sufficiently small $\delta>0$
  such that given a graph $G$ it is \NP-hard to distinguish the cases,
  \begin{description}\item[\yes:] 
    there exists a vertex set $S$ with volume $\mu(S)=\delta$ and expansion 
    $\Phi(S)\le \eta$,
  \item[\no:]  all vertex sets $S$  with volume  $\mu(S)=\delta$ have expansion
     $\Phi(S)\ge 1-\eta$.
  \end{description}
\end{hypothesis*}
For the sake of succinctness, we will refer to the above promise problem as
\smallsetexpansion with parameters $(\eta,\delta)$.  Apart from being a natural
optimization problem, the \smallsetexpansion problem is closely tied to the Unique
Games Conjecture, as discussed in the next paragraph.

Recently, Arora, Barak, and Steurer \cite{AroraBS10} showed that the
problem $\smallsetexpansion(\eta,\delta)$ admits a subexponential algorithm, namely an
algorithm that runs in time $\exp(n^\eta/\delta)$.
However, such an algorithm does not refute the hypothesis that the problem
$\smallsetexpansion(\eta,\delta)$ might be hard for every constant~$\eta>0$ and
sufficiently small $\delta>0$.

\paragraph{Unique Games Conjecture}
The Khot's Unique Games Conjecture~\cite{Khot02a} is among the
central open problems in hardness of approximation.  At the outset, the
conjecture asserts that a certain constraint satisfaction problem called
the Unique Games is hard to approximate in a strong sense.

An instance of \uniquegames consists of a graph with vertex set $V$, a finite set of
labels $[R]$, and a permutation $\pi_{v\la w}$ of the label set for each
edge $(v,w)$ of the graph.  A
labeling $F: V \to [R]$ of the vertices of the graph is said to {\it satisfy} an edge $(v,w)$,
if $\pi_{v \la w}(F(w)) = F(v)$.  The objective is to find a labeling
that satisfies the maximum number of edges.

The Unique Games Conjecture asserts that if the label set is large
enough then even though the input instance has a labeling satisfying
almost all the edges, it is \NP-hard to find a labeling satisfying
any non-negligible fraction of edges.

In recent years, Unique Games Conjecture has been shown to imply optimal
inapproximability results for classic problems like \maxcut
\cite{KhotKMO07}, \vertexcover \cite{KhotR08} \sparsestcut \cite{KhotV05}
and all constraint satisfaction problems \cite{Raghavendra08}.
Unfortunately, it is not known if the converse of any of these implications
holds.  In other words, there are no known polynomial-time {\it reductions}
from these classic optimization problems to \uniquegames, leaving the
possibility that while the its implications are true the conjecture itself
could be false.

Recent work by two of the authors established a {\it reverse}
reduction from the \smallsetexpansion problem to Unique Games
\cite{RaghavendraS10}.  More precisely, their work showed that \SSEH implies the Unique Games Conjecture.  This result suggests
that the problem of approximating expansion of small sets lies at the
combinatorial heart of the Unique Games problem.  In fact, this connection
proved useful in the development of subexponential time algorithms for
Unique Games by Arora, Barak and Steurer \cite{AroraBS10}.  It was
also conjectured in \cite{RaghavendraS10} that Unique Games Conjecture
is equivalent to the \SSEH.

\subsection{Results (Informal Description)}

In this work, we further investigate the connection between \smallsetexpansion
and the Unique Games problem.  The main result of this work is that
the \SSEH  is equivalent to a variant of the \UGCexpand.
More precisely, we show the following:

\begin{theorem*}[Main Theorem, Informal]
  The \SSEH is equivalent to assuming that the
  Unique Games Conjecture holds even when the input instances
  are required to be small set expanders, i.e., sets of roughly $\delta
  n$ vertices for some small constant $\delta$ have expansion
  close to $1$.
\end{theorem*}

As a corollary, we show that \SSEH implies hardness of
approximation results for \balancedseparator and \mla problems.  The
significance of these results stems from two main reasons.

First, the Unique Games Conjecture is not known to imply hardness results for problems closely tied to graph expansion such as \BalancedSeparator and \mla.  The reason being that the hard
instances of these problems are required to have certain global
structure namely expansion.  Gadget reductions from a unique games
instance preserve the global properties of the unique games instance
such as lack of expansion.  Therefore, showing hardness for
\BalancedSeparator or \mla problems often required a stronger version
of the \UGCexpand, where the instance is guaranteed to have good expansion.
To this end, several such variants of the conjecture for expanding graphs have
been defined in literature, some of which turned out to be false
\cite{AroraKKSTV08}.   Our main result shows that the \SSEH serves as a natural unified
assumption that yields all the implications of \UGCexpand and, in addition, also 
hardness results for other fundamental problems such as \BalancedSeparator.

Second, several results in literature point to the close connection
between \smallsetexpansion problem and the Unique Games problem.  One of the central implications of the \UGCexpand is that
certain semidefinite programs yield optimal approximation for various classes
of problems.  As it turns out, hard instances for semidefinite
programs (SDP integrality gaps) for
\maxcut \cite{FeigeS02a,KhotV05,KhotS09,RaghavendraS09c}, \vertexcover \cite{GeorgiouMPT07},
\uniquegames \cite{KhotV05,RaghavendraS09c} and
\sparsestcut \cite{KhotV05,KhotS09,RaghavendraS09c} all have near-perfect edge expansion for small
sets.  In case of \uniquegames,  not only do all known integrality gap
instances have near-perfect edge expansion of small sets, even the
analysis relies directly on this property.  
All known integrality gap instances for semidefinite programming
relaxations of Unique Games, can be translated in to gap instances for
\smallsetexpansion problem, and are arguably more natural in the latter context.
Furthermore, all the algorithmic results for \smallsetexpansion, including the
latest work of Arora, Barak and Steurer \cite{AroraBS10} extend to
Unique Games as well.  This apparent connection was formalized in the
result of Raghavendra \etal \cite{RaghavendraS10} which showed that \SSEH implies the \UGCexpand.
This work complements that of Raghavendra \etal
\cite{RaghavendraS10} in exhibiting that the \smallsetexpansion problem lies at
the combinatorial heart of the Unique Games problem.  

We also show a ``hardness amplification'' result for \smallsetexpansion proving that
if the \SSEH holds then the current best algorithm for
\smallsetexpansion due to \cite{RaghavendraS10} is optimal within some fixed constant
factor.  One can view the reduction as a ``scale change'' operation for
expansion problems, which starting from the qualitative hardness of a problem 
about expansion of sets with a sufficiently small measure $\delta$, 
gives the optimal quantitative hardness results
for problems about expansion of sets with any desired measure (larger than $\delta$).
This is analogous to (and based on) the results of \cite{KhotKMO07} who gave a
similar alphabet reduction for \uniquegames.
An interesting feature of the reductions in the paper is that they
produce instances whose expansion of small sets closely mimics a
certain graph on the Gaussian space.

\section{Preliminaries}
\label{sec:preliminaries}

\paragraph{Random walks on graphs}
\label{sec:random-walks-graphs}

Consider the natural random walk on $V$ defined by $G$.
We write $j\sim G(i)$ to denote a random neighbor of vertex $i$ in $G$
(one step of the random walk started in $i$).
The stationary measure for the random walk is given by the volume
as defined earlier with $\mu(i) = G(\{i\},V)$.
If $G$ is regular, then $\mu$ is the uniform distribution on $V$. 
In general, $\mu$ is proportional to the degrees of the vertices in~$G$.
We write $i\sim \mu$ to denote a vertex sampled according to the
stationary measure.
If $G$ is clear from the context, we often write $i\sim V$ instead of
$i\sim \mu$.

\paragraph{Spectral gap of graphs}
\label{sec:spectral-gap-graphs}

We identify $G$ with the stochastic matrix of the random walk on $G$.
We equip the vector space $\set{f\from V\to \R}$ with the inner product 
\begin{displaymath}
  \iprod{f,g}
  \defeq \E_{x\sim \mu} f(x)g(x)
  \mper
\end{displaymath}
We define $\norm{f}=\iprod{f,f}^{1/2}$.
As usual, we refer to this (Hilbert) space as $L_2(V)$.
Notice that $G$ is self-adjoint with respect to this inner product,
i.e., $\iprod{f,Gg} = \iprod{Gf,g}$ for all $f,g\in L_2(V)$.
Let $\lambda_1\ge \ldots\ge \lambda_n$ be the eigenvalues of $G$.
The non-zero constants are eigenfunctions of $G$ with eigenvalue
$\lambda_1=1$.

For a vertex set $S\sse V$, let $\Ind_S$ be the $\bits$-indicator
function of $S$.
We denote by $G(S,T) = \iprod{\Ind_S,G \Ind_T}$ the total weight of
all the edges in $G$ that go between $S$ and $T$.
\begin{fact}
\label{fact:second-eigenvalue}
  Suppose the second largest eigenvalue of $G$ is $\lambda$.
  Then, for every function $f\in L_2(V)$,
  \begin{displaymath}
    \iprod{f,Gf } \le (\E f)^2 + \lambda\cdot \Paren{\snorm f - (\E f)^2}
    \mper
  \end{displaymath}
  In particular, $\Phi_G(\delta)\ge 1-\delta-\lambda$ for every
  $\delta>0$.
\end{fact}

\paragraph{Gaussian Graphs}
For a constant $\rho \in (-1,1)$, let $\cG(\rho)$ denote the infinite graph
over $\R$ where the weight of an edge $(x,y)$ is the probability that two
standard Gaussian random variables $X,Y$ with correlation $\rho$ equal $x$
and $y$ respectively.
The expansion profile of Gaussian graphs is given by
$\cond_{\cG(\rho)}(\mu) = 1 - \Gamma_{\rho}(\mu)/\mu$ where the quantity
$\Gamma_{\rho}(\mu)$ defined as
\begin{displaymath}
  \Gamma_\rho(\mu)
  \seteq\Prob[(x,y)~\cG_\rho]{x\ge t, y\ge t}\mcom
\end{displaymath}
where
$\cG_\rho$ is the $2$-dimensional Gaussian distribution with
covariance matrix 
\begin{displaymath}
    \Paren{
      \begin{matrix}
        1 & \rho \\
        \rho & 1
      \end{matrix}}
\end{displaymath}
and $t\ge 0$ is such that $\Prob[(x,y)\sim\cG_\rho]{x\ge t}=\mu$.  
A theorem of Borell~\cite{Borell85} shows that for any set $S$ of measure
$\mu$, $(\cG(\rho)) (S,S) \leq \Gamma_{\rho}(\mu)$.  This expansion profile
will be frequently used in the paper to state the results succinctly.

\paragraph{Noise graphs}
\label{sec:noise-graphs}

For a finite probability space $(\Omega,\nu)$ and $\rho\in[0,1]$, we
define $T=T_{\rho,\Omega}$ to be the following linear operator on
$L_2(\Omega)$,
\begin{displaymath}
  T f(x) 
  = \rho x + (1-\rho) \E_{x'\sim \Omega} f(x')
  \mper
\end{displaymath}

The eigenvalues of $T$ are $1$ (with multiplicity $1$) and $\rho$
(with multiplicity $\card{\Omega}-1$).
The operator $T$ corresponds to the following natural
(reversible) random walk on $\Omega$: with probability $\rho$ stay at
the current position, with probability $(1-\rho)$ move to a random
position sampled according to the measure $\nu$.

\paragraph{Product graphs}
\label{sec:product-graphs}

If $G$ and $G'$ are two graphs with vertex sets $V$ and $V'$, we let
$H=G\tensor G'$ be the \mbox{\emph{tensor product}} of $G$ and $G'$.
The vertex set of $G$ is $V\times V'$.
For $i\in V$ and $i'\in V'$, the distribution $H(i,i')$
is the product of the distributions $G(i)$ and $G'(i')$.
For $R\in\N$, we let $G^{\tensor R}$ denote the \emph{$R$-fold tensor
  product} of~$G$.
Sometimes the power $R$ of the graph is clear from the context.
In this case, we might drop the superscript for the tensor graph.

\section{Results}
Towards stating the results succinctly, we introduce the notion of a
decision problem being \ssehard.  It is the natural notion wherein a
decision problem is \ssehard if the \smallsetexpansion$(\eta,\delta)$ reduces to it by
a polynomial time reduction for some constant $\eta$ and all $\delta>0$
(See \pref{def:ssehard}).

\subsection{Relation to the Unique Games Conjecture}

We show that the \SSEH  is equivalent to a
certain variant of the \UGCexpand with expansion.  Specifically, consider the
following version of the conjecture with near-perfect expansion of sufficiently
small sets.  The hypothesis is as follows: \footnote{%
  The hypothesis in \cite{RaghavendraS10} is not quite the same.
  However, the reduction and its analysis in \cite{RaghavendraS10} also
  work for this
  hypothesis.  
}

\begin{hypothesis}[Unique Games with Small-Set Expansion]
\label{hyp:expanding-unique-games}
  For every $\e,\eta>0$ and $M\in \N$, there exists
  $\delta=\delta(\e,M)>0$ and $q=q(\e,\eta,M)\in\N$ such that it is \nphard to distinguish for a given \uniquegames instance
  $\UG$ with alphabet size $q$ whether
  \begin{description}
  \item[\yes:] The \uniquegames instance $\UG$ is almost
    satisfiable, $\opt(\UG)>1-\e$.
  \item[\no:] The \uniquegames instance $\UG$ satisfies
    $\opt(\UG)<\eta$ and its constraint graph $G$ satisfies
    $\Phi(S)>1-\e$ for every vertex set with $\delta \le \mu(S)\le
    M\delta$.  
\Pnote{this looks like a neat way to state it, since we want both upper and lower bounds on set size}
  \end{description}  
\end{hypothesis}

The main result of the paper is the following reduction from
\smallsetexpansion to \uniquegames on instances with small-set expansion.
\begin{theorem}
\torestate{
\label{thm:expanding-unique-games}
  For every $q\in \N$ and every $\e,\gamma>0$, it is \ssehard to
  distinguish between the following cases for a given \uniquegames
  instance $\UG$ with alphabet size $q$:
  \begin{description}
  \item[\yes:] The \uniquegames instance $\UG$ is almost
    satisfiable, $\opt(\UG)>1-2\e - o(\e)$
  \item[\no:] The optimum of the \uniquegames instance $\UG$ is negligible, and the expansion profile of the instance resembles the Gaussian graph $\cG(1-\e)$.  More precisely, the \uniquegames instance $\UG$ satisfies
    $\opt(\UG)<O\left({q^{-\ffrac{\e}{(2-\e)}}}\right)+\gamma$ and in addition,
    the constraint graph $G$ of $\UG$ satisfies
    \begin{displaymath}
      \forall S\sse V(G).\qquad \cond_G(S) \geq
      \cond_{\cG(1-\e)}\bigparen{\vol(S)} -\nfrac{\gamma}{\vol(S)}
      \mper
    \end{displaymath}
  \end{description}
}
\end{theorem}
The proof of the above theorem is presented in \pref{sec:putt-things-together}.  
Together with
Theorem 1.9 from \cite{RaghavendraS10}, \pref{thm:expanding-unique-games}  implies the following
equivalence:
\begin{corollary}
  The \SSEH is equivalent to
  \pref{hyp:expanding-unique-games} (Unique Games with
  Small-Set Expansion).
\end{corollary}

\begin{remark}
If we choose $\gamma\ll \e$, then the constraint graph $G$
in the \no case satisfies $\Phi(S)\ge \Omega(\sqrt \e)$ for every
vertex set $S$ with $\mu(S) \in (b,\half)$ for an arbitrarily small
constant $b>0$.  In other words, the best balanced separator in $G$ has cost
$\Omega(\sqrt \e)$.  A hardness of Unique Games on graphs of this
nature was previously conjectured in \cite{AroraKKSTV08}, towards
obtaining a hardness for \BalancedSeparator.
\end{remark}

As already mentioned, for several problems such as \maxcut, the
the hard instances for the semidefinite programs have very good
expansion of small sets.  For instance, hard instances for semidefinite
programs (SDP integrality gaps) for
\maxcut \cite{FeigeS02a,KhotV05,KhotS09,RaghavendraS09c}, \vertexcover \cite{GeorgiouMPT07},
\uniquegames \cite{KhotV05,RaghavendraS09c} and
\sparsestcut \cite{KhotV05,KhotS09,RaghavendraS09c} all have near-perfect edge expansion for small
sets.  In fact, in many of the cases, the edge expansion in the graph closely
mimics the expansion of sets in some corresponding Gaussian graph.  
Confirming this observation, our techniques imply an optimal hardness result for \maxcut on instances that are small-set expanders.  More precisely, the \SSEH implies that the Goemans-Williamson algorithm is optimal even on graphs that are guaranteed to have good expansion of small sets, in fact an expansion profile that resembles the Gaussian graph.  For the sake of succinctness, we omit the formal statement of the result.

\subsection{Hardness Amplification for Graph Expansion}

Observe that the \SSEH is a purely qualitative assumption on the approximability of expansion.  Specifically, for every constant $\eta$ the hypothesis asserts that there exists some $\delta$ such that approximating expansion of sets of size $\delta$ is \NP-hard.  
The hypothesis does not assert any quantitative dependence on the set size and approximability.  Surprisingly, we show that this qualitative hardness assumption is sufficient to imply precise quantitative bounds on approximability of graph expansion.

\begin{theorem}
\torestate{
\label{thm:general-sse}
For all $q \in \N$ and $\e,\gamma > 0$, it is \ssehard to distinguish between
the following two cases for a given graph $H = (V_H,E_H)$
\begin{description}
 \item[\yes:]
There exist $q$ disjoint sets $S_1, \ldots, S_q \sse V_H$ satisfying for all $l \in [q]$,
\[ \vol(S_l) = \tfrac1q \qquad \text{and} 
\qquad \cond_{H}(S_l) \leq \epsilon + o(\epsilon).\]
\item[\no:] For all sets $S \sse V_H$,
\[ \cond_{H}(S) ~\geq~ \cond_{\cG(1-\epsilon/2)}\bigparen{\vol(S)}
-\nfrac{\gamma}{\vol(S)}
\]
where $\cond_{\cG(1-\epsilon/2)}(\vol(S))$ is the expansion of sets of volume $\vol(S)$ in the infinite Gaussian graph $\cG(1-\epsilon/2)$.
\end{description}
}
\end{theorem}

The above hardness result matches (upto an absolute constant
factor), the recent algorithmic result (Theorem 1.2) of
\cite{RaghavendraST10} approximating the graph expansion.
Furthermore, both the \yes  and the \no cases of the above theorem are even qualitatively stronger than in the 
\SSEH. In the \yes case, not only does the graph have one non-expanding set,
but it can be partitioned into small sets, \emph{all} of which are non-expanding. This partition
property is useful in some applications such as hardness reduction to \mla.
In the \no case, the expansion of all sets can be characterized only by their size $\mu(S)$.  Specifically, the expansion of every set $S$ of vertices with $\vol(S) >> \gamma$, is at least the expansion of a set of similar size in the Gaussian graph $\cG(1-\epsilon/2)$.

Here we wish to draw an analogy to the Unique Games Conjecture.  The Unique
Games Conjecture is qualitative in that it does not prescribe a relation
between its soundness and alphabet size.  However,  the work of Khot \etal
\cite{KhotKMO07} showed that the \UGCexpand implies a quantitative form of itself with a precise relation between the alphabet size and soundness.  Theorem \ref{thm:general-sse} could be thought of as an analogue of this phenomena for the Small-Set Expansion problem.

As an immediate consequence of Theorem \ref{thm:general-sse}, we
obtain the following hardness of the \BalancedSeparator and \mla
problems (See Appendix \ref{app:mlabalanced} for details).

\begin{corollary}[Hardness of \BalancedSeparator and \minbisection]
\label{cor:balancedseparator}
There is a constant $c$ such that for arbitrarily small $\epsilon > 0$, it is 
\ssehard to distinguish the following two cases 
for a given graph $G=(V,E)$:
\begin{description}
 \item[\yes:] There exists a cut $(S,V \setminus S)$ in $G$ such that 
$\vol(S) = \tfrac12$ and $\cond_G(S) \leq \epsilon + o(\epsilon)$.
 \item[\no:] Every cut $(S,V\setminus S)$ in $G$, with 
$\vol(S) \in \inparen{\tfrac{1}{10},\tfrac12}$ satisfies 
$\cond_G(S) \geq c \sqrt{\epsilon} $.
\end{description}
\end{corollary}
\begin{corollary}[Hardness of \mla] \label{cor:mla} 
It is \ssehard to approximate \mla to any fixed constant factor.
\end{corollary}

\section{Warm-up: Hardness for Balanced Separator}
\label{sec:hardn-balanc-separ}
In this section we present a simplified version of our reduction from \SSE to
\BalancedSeparator. Though it gives sub-optimal parameters, it illustrates the
key ideas used in the general reduction.

\subsection{Candidate Reduction from Unique Games}
\label{sec:cand-reduct-from}

A natural approach for reducing \uniquegames to \balancedseparator is to
consider variants of the reduction from \uniquegames to \maxcut in
\cite{KhotKMO07} (similarly, one could consider variants of the reduction
from \uniquegames to the \emph{generalized} \sparsestcut problem \cite{KhotV05}).

Let $\UG$ be a unique game with alphabet size $R$ and vertex set $V$.
(We assume that every vertex of the unique game participates in the same
number of constraints.
This assumption is without loss of generality.)
The candidate reduction has a parameter $\e>0$.
The graph $H=H_\e(\UG)$ obtained from this candidate reduction has vertex
set $V\times \bits^R$ and its edge distribution is defined as follows:
\begin{enumerate}
\item Sample a random vertex $u\in V$.

\item Sample two random constraints $(u,v,\pi),(u,v',\pi')$ of $\UG$ that
  contain the vertex $u$.
  (Henceforth, we will write $(u,v,\pi)\sim \UG\mid u$ to denote a random
  constraint of $\UG$ containing vertex $u$.)
\item Sample a random edge $(y,y')$ of the boolean noise graph $T_{1-\e}$
  with noise parameter~$\e$.

\item Output an edge between $(v,\pi(y))$ and $(v',\pi'(y'))$.
  (Here, $\pi(y)$ denotes the vector obtained by permuting the coordinates
  of $y$ according to the permutation $\pi$.)
\end{enumerate}

\paragraph{Completeness}
\label{sec:completeness-1}

Suppose there is a good assignment $F \from V\to [R]$ for the unique game
$\UG$.
Then, if we sample a random vertex $u\in V$ and two random constraint
$(u,v,\pi),(u,v',\pi')\sim \UG\mid u$, with probability very close to $1$
(much closer than $\e$), the labels assigned to $v$ and $v'$ satisfy
$\pi^{-1}(F(v)) = (\pi')^{-1}(F(v))$.
Consider the vertex set
\begin{math}
  S = \set{ (u,x) \mid x_{F(u)}=1}\mper
\end{math}
in the graph $H$. 
We have $\mu(S) = 1/2$.
We claim that the expansion of this set is essentially $\e/2$ (up to a
lower-order term depending on the fraction of constraint of $\UG$
violated by $F$).
Consider a random edge $e$ with endpoints $(v,\pi(y))$ and $(v',\pi'(y'))$, where
the vertices $v,v'\in V$ and the permutations $\pi,\pi'$ are generated as
specified above.
Let $r=\pi^{-1}(F(v))$ and $r'=(\pi')^{-1}(F(v))$.
The edge $e$ crosses the cut $S$ if and only if $y_r \neq {y'}_{r'}$.
As argued before, with probability very close to $1$, we have $r=r'$.
Conditioned on this event, the probability that $y_r\neq y_{r'}$ is equal
to $\e/2$.
This shows that $S$ has expansion $\e/2$.

\paragraph{Soundness}
\label{sec:soundness-1}

Suppose no assignment for the unique game $\UG$ satisfies a significant
fraction of constraints.
Let $S$ be a vertex set in the graph $H$. 
The goal is to lower bound the expansion of $S$ (which is the same as upper
bounding the fraction of edges with both endpoints in $S$).
Let $f \from V^R \times \bits^R \to \bits$ be the indicator function of $S$.
Following the analysis of \cite{KhotKMO07}, we consider functions $g_u\from
\bits^R\to [0,1]$,
\begin{displaymath}
  g_u(x) = \Prob[\substack{%
    (u,v,\pi)\sim \UG \mid u,\\
    y\sim T_{\sqrt{1-\e}}(x)
  }]{%
    f(v,\pi(y))
  }
  \mper
\end{displaymath}
(The graph $H$ turns out to be the square of a graph $H_0$
in which we would just create edges of the form $((u,x), (v,\pi(y)))$.
The function $g_u(x)$ evaluates to the probability that the set~$S$
contains a random neighbor of $(u,x)$ in this graph $H_0$.)
By construction, the fraction $H(S,S)$ of edges of $H$ with both endpoints
in $S$ is exactly
\begin{displaymath}
  H(S,S) = \E_{u\in V} \iprod{g_u,T_{1-\e}g_u}\mper
\end{displaymath}
Since $\UG$ does not have a good assignment, standard arguments (invariance
principle and influence decoding, see \cite{KhotKMO07}) imply the following
upper bound on $H(S,S)$,
\begin{displaymath}
  H(S,S) \le \E_{u\in V} \Gamma_{1-\e}(\mu_u) + o(1)
  \mper
\end{displaymath}
(The notation $o(1)$ hides a term depending on the maximum fraction of
constraints of $\UG$ that can be satisfied. 
For us, this term is not significant.)
Here, $\mu_u$ is the expected value of $g_u$ and $\Gamma_{1-\e}(\cdot )$ is
the noise stability profile of the Gaussian noise graph with parameter
$\e$.
\Mnote{Is there a better way of describing $\Gamma_{1-\e}$. In my Princeton and IAS some people seemed
  confused by the phrase ``Gaussian noise graph''.}
We would like to show that every set $S$ that contains a $\mu$ fraction of
the vertices of $H$ satisfies $H(S,S)\le \Gamma_{1-\e}(\mu)+o(1)$.
However, the function $\Gamma_{1-\e}$ is, of course, not concave.
Hence, this upper bound holds only if $\mu_u$ is close to $\mu$ for most
vertices~$u\in V$.
%

In fact, it is very easy to construct examples that show that the candidate
reduction is not sound.
For example, consider a unique game $\UG$ that consists of two disjoint
parts of the same size (i.e., without any constraint between the two
parts).
The reduction preserves this global structure,
in the sense that the graph~$H$ also consists of two disjoint parts of the
same size (with no edge between the parts).
Hence, this graph contains a vertex set with volume~$\half$ and
expansion~$0$ irrespective of the optimal value of the unique game~$\UG$.
In fact, any cut in the underlying graph of $\UG$ can be translated to a cut in $H$
and the resulting function $f$ may have the values $\mu_u$ as
(very close to) 0 or 1.

This example shows that the above candidate reduction can only work if one
makes assumptions about structure of the constraint graph of the underlying
unique game~$\UG$. 
However, such an assumption raises the question if \uniquegames could be
hard to approximate even if the constraint graph is expanding.
This issue turns out to be delicate as demonstrated by the algorithm for
\uniquegames with expanding constraint graphs \cite{AroraKKSTV08}.
This algorithm achieves a good approximation for \uniquegames if the
expansion of the constraint graph exceeds a certain threshold.

\subsection{Structured Unique Games from Small-Set Expansion}
\label{sec:struct-uniq-games}

In this work, we present a very different approach for fixing the above
candidate reduction.
Instead of assuming expansion properties of the constraint graph, we assume
that the underlying unique game is obtained by the reduction
from \smallsetexpansion to \uniquegames in \cite{RaghavendraS10}
\footnote{We remark that unique games of this form do not necessarily have
expanding constraint graphs.
In fact, it is still possible that the constraint graph consists of two
disconnected components.}.
This specific form of the underlying unique game will allows us to modify
the reduction such that the global structure of the constraint graph is no
longer preserved in the graph obtained from the reduction.
(In particular, our modified reduction will break with the paradigm of
composing unique games with local gadgets.)

In the following, we describe the reduction from \smallsetexpansion to
\uniquegames.
Let $G$ be a regular graph with vertex set $V$.
For technical reasons, we assume that $G$ contains a copy of the complete
graph of weight $\eta>0$. 
(Since we will be able to work with very small $\eta$, this assumption is
without loss of generality.)
Given a parameter $R\in \N$ and the graph $G$, the reduction outputs a
unique game $\UG=\UG_R(G)$ with vertex set $V^R$ and alphabet $[R]$.
The constraints of the unique game $\UG$ correspond to the following
probabilistic verifier for an assignment $F\from V^R\to [R]$:
\begin{enumerate}
\item Sample a random vertex $A\in V^R$.

\item Sample two random neighbors $B,C\sim G^{\tensor R}(A)$ of the vertex
  $A$ in the tensor-product graph $G^{\tensor R}$.

\item Sample two random permutations $\pi_B,\pi_C$ of $[R]$.

\item Verify that $\pi_B^{-1}(F(\pi_B(B)))=(\pi_C)^{-1}(F(\pi_C(C)))$.
\end{enumerate}

Raghavendra and Steurer \cite{RaghavendraS10} show that this reduction is
complete and sound in the following sense:
\begin{description}
\item[Completeness] If the graph $G$ contains a vertex set with volume
  $1/R$ and expansion close to $0$, then the unique game $\UG=\UG_R(G)$ has
  a partial assignment that labels an $\alpha\ge 1/e$ fraction of the
  vertices and satisfies almost an $\alpha$ fraction of the constraints.
\item[Soundness] If the graph $G$ contains no set with volume $1/R$ and
  expansion bounded away from $1$, then no assignment for the unique game
  $\UG=\UG_R(G)$ satisfies a significant fraction of the constraints.
\end{description}

Hence, if one assumes the Small-Set Expansion Hypothesis, then the kind of
unique games obtained from the reduction are hard to approximate.

We remark that the completeness of the reduction seems weaker than usual,
because we are only guaranteed a partial assignment for the unique game.
However, it is easy to check that the KKMO reduction presented in the
previous section also works if there is only a partial assignment in the
completeness case.
The only difference is that one now gets a set $S$ with $\mu(S) = \alpha/2$ and expansion
roughly $\e/2$.

\subsection{Reduction from Small-Set Expansion to Balanced Separator}
\label{sec:reduction-from-small}

We now show how the combination of the above two reductions can be modified to give a reduction from
\SSE to \BalancedSeparator.
Let $\UG = \UG_R(G)$ be the unique game given by the reduction of Raghavendra and Steurer. If we consider the
graph $H_{\e}$ given by the reduction in \pref{sec:cand-reduct-from}, each vertex of $H_{\e}$ is now
of the form $(A,x)$, where $A \in V^R$ and $x \in \B^R$.

The intuition is that in this case, we can think of $x$ as picking a \emph{subset} of the
vertices in $A$, and that just the knowledge of this subset (instead of the whole of $A$) 
is sufficient for the provers to provide a good answer to the corresponding unique game. 
In particular, let $A' = \set{A_i \mid x_i = 1}$ is the subset picked by $x$. 
Then the argument for the completeness case in \cite{RaghavendraS10} actually shows 
that one can still find a good labeling for an $\alpha$ fraction of the vertices $A$, 
where the label of $A$ only depends on $A'$
\footnote{Given a non-expanding small set $S$, if $A' \cap S$ contains a single
  element $A_j'$, then we assign the label $j$ to $A$. If $A' \cap S$ is not a singleton, we do not
  label $A$.}.

Formally, if we replace $A$ with the tuple $A'(x)$ defined by taking $A_i' = A_i$ if $x_i = 1$ and $A_i' =
\bot$ otherwise.
This gives a graph $H'$ with the vertex set being a subset of $(V \cup \{\bot\})^R \times \B^R$. 
The the argument in completeness case 
for showing that $H$ has a balanced cut of expansion 
roughly $\e/2$ can in fact be extended to show that $H'$ also has a balanced  cut of expansion
roughly $\e/2$.

The soundness analysis in the previous reduction did not always work because $H$ had the same
structure as $G^{\tensor R}$, since we essentially replaced every vertex of $G^{\tensor R}$ by a gadget $\B^R$ to obtain
$H$. 
However, the structure of $H'$ is very different from that of $G^{\tensor R}$.

For example, consider the vertices $A = (u_1, u_2,\ldots, u_R)$ and $B = (v_1, u_2, \ldots u_R)$ in
$V^R$ which only differ in the first coordinate ($A,B$ are not necessarily adjacent). Let $x \in
\B^R$ be such that $x_1 = 0$. 
Then, while $(A,x)$ and $(B,x)$ are different vertices in $H$, 
$(A'(x),x)$ and $(B'(x),x)$ are in fact the same vertex in $H'$!
On the other hand, if $x_1 = 1$, then $(A'(x),x)$ and $(B'(x),x)$ would be two different vertices in
$H'$.
Hence, the gadget structure of $H$ is no longer preserved in $H'$ - it is very different from a
``locally modified'' copy of $G^{\tensor R}$. 

For the purposes of analysis, it will be more convenient to think of $A'$ being obtained
by replacing $A_i$ where $x_i = 0$, by a random vertex of $G$ instead of the symbol $\bot$.
Instead of identifying different vertices in $H$ with the same vertex in $H'$, this now has the
effect of re-distributing the weight of an incident on $(A,x)$, uniformly over all the vertices that 
$(A',x)$ can map to.
Let $M_{x}$ denote a Markov operator which maps $A$ to a random $A'$ as above
(a more general version and analysis of such operators can be found in \pref{sec:additional-prelims}).

We now state the combined reduction. The weight of an edge in the final graph $H'$ is the probability
that it is produced by the following process:
\begin{enumerate}
\item Sample a random vertex $A\in V^R$.
\item Sample two random neighbors $B,C\sim G^{\tensor R}(A)$ of the vertex
  $A$ in the tensor-product graph $G^{\tensor R}$.
\item Sample $x_B, x_C \sim \B^R$.
\item Sample $B' \sim M_{x_B}(B)$ and $C' \sim M_{x_C}(C)$.
\item Sample two random permutations $\pi_B,\pi_C$ of $[R]$.
\item Output an edge between the vertices $\pi_B(B', x_B)$ and $\pi_C(C', x_C)$~ ($\pi(A,
  x)$ denotes the tuple $(\pi(A), \pi(x))$).
\end{enumerate}

As before, let $f: V^R \times \bits^R$ denote the indicator function of a set in $H'$, with (say)
$\E f = \mu = 1/2$. We define the functions
\[ \bar{f} (A,x) \defeq \E_{\pi} f(\pi.A,\pi.x) \qquad \text{and} 
\qquad g_A(x) \defeq \E_{B \sim G^{\tensor R}(A) } \E_{B' \sim M_x(B)} {\bar f}(B',x)\mper\]

By construction, each vertex $(A,x)$ of $H'$ has exactly the same 
neighborhood structure as  $(\pi.A', \pi.x)$ for all $\pi \in S_R$ and $A' \in M_x(A)$. 
Hence, the fraction of edges crossing the cut can also be written in terms of $\bar f$ as
$\iprod{f, H' f} = \iprod{\bar f, H' \bar f}$.

We will show that $\bar f$ gives a cut (actually, a distribution over cuts) with the same expansion
in the graph $H$, such that the functions $g_A$ satisfy 
$\Prob[A]{\E g_A \in  (\nfrac{1}{10},\nfrac{9}{10})} \geq \nfrac{1}{10}$. 
Recall that showing this was exactly the problem
in making the reduction in \pref{sec:cand-reduct-from} work. 

Since $\E_A \E_x {g_A} = \mu$, we have $\E_A \inparen{\E_x g_A}^2 \geq \mu^2$. The following claim
also gives an upper bound.
\begin{claim}
$\E_A \inparen{\E_x g_A}^2 \leq \mu^2/2 + \mu/2$
\end{claim}
\begin{proof}
We have
\begin{align*}
\E_{A \sim V^R} \inparen{\E_x g_A}^2  
= \E_{A \sim V^R} \inparen{\E_{B \sim G^{\tensor R} (A)} \E_x \E_{B' \sim M_x(B)} {\bar f}}^2
&\leq~ \E_{A \sim V^R} \E_{B \sim G^{\tensor R}} \inparen{\E_x \E_{B' \sim M_x(B)} {\bar f}}^2\\
&=~ \E_{B \sim V^R} \inparen{\E_x \E_{B' \sim M_x(B)} {\bar f}}^2\\
&=~ \Ex[B \sim V^R]{\inparen{\E_{x_1} \E_{B_1' \sim M_{x_1}(B)} {\bar f}}\inparen{\E_{x_2} \E_{B_2'
      \sim M_{x_2}(B)} {\bar f}}} \\
&=~ \E_{x_1} \E_{B_1' \sim M_{x_1}(B)} {\bar f}(B_1',x_1) \E_{(B_2',x_2) \sim M (B_1',x_1)}
f(B_2',x_2) \mper
\end{align*}
For the last equality above, we define $M$ to be a Markov operator which samples $(B_2',x_2)$
from the correct distribution given $(B_1',x_1')$. 
Since $x_1,x_2$ are independent, $x_2$ can just
be sampled uniformly. 
The fact that $B_1'$ and $B_2'$ come from the same (random) $B$ can be
captured by sampling each coordinate of $B_2'$ as
\[
(B_2')_i ~=~ \left\{ \begin{array}{ll}
(B_1')_i & \text{if}~ (x_1)_i = (x_2)_i = 0\\
\text{random vertex in}~ V & \text{otherwise}  
\end{array} \right. \mper
\]
Abusing notation, we also use $M$ to denote the operator on the space of the functions which
averages the value of the function over random  $(B_2',x_2)$ generated as above. 
Then, if $\lambda$ is the second eigenvalue of $M$,  we have
\[\E_A \inparen{\E_x g_A}^2 
~\leq~ \iprod{\bar f, M \bar f} 
~\leq~ 1 \cdot (\E \bar f)^2 + \lambda \cdot \inparen{\norm{\bar f}^2 - (\E \bar f)^2}
~\leq~ (1 - \lambda) \cdot \mu^2 + \lambda \cdot \mu \mper\]
Finally, it can be checked that the second eigenvalue of $M$ is 1/2 which proves the claim.
\end{proof}
This gives that $\E_x g_A$ cannot be always very far from $\mu$. Formally,
\[ \Prob[A]{\abs{\E g_A - \mu} \geq \gamma} 
~\leq~ \frac{\E_A (\E g_A - \mu)^2}{\gamma^2}
~\leq~ \frac{\mu(1-\mu)}{2\gamma^2} \mper\]
Hence, for $\gamma = 2/5$, the probability is at most $\nfrac{25}{32} < \nfrac{9}{10}$.
This can now be combined with the bound from \pref{sec:cand-reduct-from} that gives
\[
H'(S,S)  \leq \E_{A} \Gamma_{1-\e} \inparen{\E g_A} + o(1) \mper
\]
Since $\E g_A \geq \nfrac{1}{10}$ with probability at least $\nfrac{1}{10}$ over $A$, these 
``nice'' $A$'s contribute a volume of at least $\nfrac{1}{100}$. Also, for a nice $A$, we have 
$\Gamma_{1-\e} \inparen{\E g_A} \leq (\E g_A)(1- \Omega(\sqrt{\e}))$. Hence,
\[
H'(S,S)  \leq (\mu-\nfrac{1}{100}) + \nfrac{1}{100} \cdot (1 - \Omega(\sqrt{\e})) + o(1)
\]
which shows that $S$ has expansion $\Omega(\sqrt{\e})$.

\section{Additional Preliminaries}
\label{sec:additional-prelims}

\paragraph{Unique Games}
  An instance of \uniquegames represented as $\uginst =
  ( \cV,\cE,\Pi,[R])$ consists of a
graph over vertex set $\cV$ with the edges $\cE$
  between them.  Also part of the instance is a set of labels $[R] =
  \{1,\ldots,R \}$, and a set of permutations $\Pi = \{ \pi_{v \la w} : [R]
  \rightarrow [R]\}$, one permutation for each edge $e = (w,v) \in \cE$.  An
assignment $F \from \cV \to [R]$
  of labels to vertices is said to satisfy an edge $e = (w,v)$, if
  $\pi_{v \la w}(F) = F(v)$.  The objective is to find an
assignment $F$ of labels that satisfies the maximum number of edges.

As is customary in hardness of approximation, one defines a
gap-version of the \uniquegames problem as follows:
\begin{problem}[\uniquegames$(R, 1-\ugcomp, \ugsound)$]
Given a \uniquegames instance $\uginst =
( \cV,\cE, \Pi = \{\pi_{v \la
w}:[R]\rightarrow [R] ~ |  ~ e=(w,v) \in \cE\},[R])$ with number of
labels $R$, distinguish between the following two cases:
\begin{itemize} \itemsep=0ex
\item {\sf $(1-\ugcomp)$- satisfiable instances:} There exists an
	assignment $F$ of labels that satisfies a $1 - \ugcomp$ fraction
	of edges.
\item {\sf Instances that are not $\ugsound$-satisfiable:}  No
assignment satisfies more than a $\ugsound$-fraction of the edges $\cE$.
\end{itemize}
\end{problem}
\Mnote{Changed the notation for assignment to $F$ instead of $\cA$ to be consistent 
with later sections.}
The Unique Games Conjecture asserts that the above decision problem is
\NP-hard when the number of labels is large enough.  Formally,
\begin{conjecture}[Unique Games Conjecture \cite{Khot02a}]
For all constants $\ugcomp,\ugsound > 0$, there exists large enough constant
$R$ such that \uniquegames$(R,1-\ugcomp,\ugsound)$ is \NP-hard.
\end{conjecture}

\paragraph{Graph expansion}
\label{sec:graphs}

In this work, all graphs are undirected and possibly weighted.
Let $G$ be a graph with vertex set $V$.
We write $ij\sim G$ to denote a random edge sampled from $G$ (with
random orientation).
For two vertex sets $S,T\sse V$, let $G(S,T)$ be the fraction of edges going from $S$ to $T$, i.e.,
\begin{displaymath}
  G(S,T)
  \defeq \Prob[ij\sim G]{i\in S, j\in T}
  \mper
\end{displaymath}
The \emph{expansion}\footnote{The technically more precise term is
  \emph{conductance} }
$\cond_G(S)$ of a set $S\sse V$ is the fraction of edges leaving $S$
normalized by the fraction of edges incident to $S$, i.e.,
\begin{displaymath}
  \cond_G(S) 
  \defeq \frac{G(S,V\sm S)}{G(S,V)}
  = \frac{\Prob[ij \sim G]{i \in S, j \notin T}}{\Prob[ij \sim G]{i \in S}}
  = \Prob[ij \sim G]{j \notin T \given i \in S}
  \mper
\end{displaymath}
The \emph{volume} of a set $S$ is the fraction of edges incident on it
and is denoted by $\vol(S) \defeq G(S,V)$. The fraction of edges leaving the set 
is denoted by $\surf(S)\defeq G(S,V\sm S)$.

\paragraph{\SSEH}
\begin{problem}[\smallsetexpansion$(\eta,\delta)$] 
Given a regular graph $G=(V,E)$, distinguish between the following two cases:
\begin{description}
\item[\yes:] There exists a non-expanding set $S \sse V$ with $\mu(S) = \delta$ 
and $\Phi_G(S) \leq \eta$.
\item[\no:] All sets $S \sse V$ with $\mu(S) = \delta$ are highly expanding having 
$\Phi_G(S) \geq 1-\eta$.
\end{description}
\end{problem}

\begin{hypothesis}[Hardness of approximating \smallsetexpansion]
For all $\eta > 0$, there exists $\delta > 0$ such that the promise problem
\smallsetexpansion($\eta,\delta$) is {\sf NP}-hard.
\end{hypothesis}\label{hypo:sse-hardness}

\begin{remark}
It is easy to see that for the problem \smallsetexpansion($\eta,\delta$) to be hard, one must have
$\delta \leq \eta$. This follows from the fact that if we randomly sample a set $S$ containing
a $\delta$ fraction of the vertices (and hence, having volume $\delta$ for a regular graph), the
expected fraction of edges crossing the set is $\delta(1-\delta)$ and hence $\E\cond_G(S) = 1-\delta$.
However, for it to be possible that for all sets with $\vol(S) = \delta$ have $\cond_G(S) \geq 1-\eta$, 
we must have $\delta \leq \eta$.
\end{remark}

\begin{definition} \label{def:ssehard}
Let $\cP$ be a decision problem of distinguishing between two disjoint
families (cases) of instances denoted by $\{\yes,\no\}$. For a given instance
$\cI$ of $\cP$, let $\case(\cI)$ denote the family to which $\cI$ belongs.
We say that $\cP$ is \ssehard if for some $\eta > 0$ and
all $\delta \in (0,\eta)$, there is a polynomial time reduction, which starting from an
instance $G=(V,E)$ of $\smallsetexpansion(\eta,\delta)$, produces an instance $\cI$ of $\cP$
such that
\begin{itemize}
 \item $\exists S \sse V$ with $\vol(S) = \delta$ and $\cond_G(S) \leq \eta$
$\quad\implies\quad$ $\case(\cI) = \yes$.
\item $\forall S \sse V$ with $\vol(S) = \delta$, $\cond_G(S) \geq 1-\eta$
$\quad \implies \quad$ $\case(\cI) = \no$.
\end{itemize}
\end{definition}

For the proofs, it shall be more convenient to use the following version of the $\smallsetexpansion$ problem,
in which we high expansion is guaranteed not only for sets of measure $\delta$, but also within 
an arbitrary multiplicative factor of $\delta$.

\begin{problem}[\smallsetexpansion$(\eta,\delta,M)$] 
Given a regular graph $G=(V,E)$, distinguish between the following two cases:
\begin{description}
\item[\yes:] There exists a non-expanding set $S \sse V$ with $\mu(S) = \delta$ 
and $\Phi_G(S) \leq \eta$.
\item[\no:] All sets $S \sse V$ with $\mu(S) \in \inparen{\tfrac{\delta}{M},M\delta}$ have
$\Phi_G(S) \geq 1-\eta$.
\end{description}
\end{problem}

The following proposition shows that for the purposes of showing that $\cP$ is \ssehard,
it is sufficient to give a reduction from \smallsetexpansion($\eta,\delta,M$) for any chosen values of
$\eta, M$ and for all $\delta$. We defer the proof to \pref{sec:stronger-small-set}.

\begin{proposition}
\torestate{
\label{prop:stronger-sse}
For all $\eta > 0, M \geq 1$ and all $\delta < 1/M$, there is polynomial time reduction
from \smallsetexpansion$(\tfrac{\eta}{M},\delta)$ to $\smallsetexpansion(\eta,\delta,M)$.
}
\end{proposition}

\paragraph{Invariance principle}
\label{sec:invariance-principle}

The following theorem on the noise stability of functions over a
product probability space is an easy corollary of
Theorem~$4.4$ in Mossel et al.~\cite{MosselOO05}.
Recall that $\Gamma_\rho(\mu)\seteq\Prob[(x,y)~\cG_\rho]{x\ge t, y\ge t}$, where
$\cG_\rho$ is the $2$-dimensional Gaussian distribution with
covariance matrix 
{\footnotesize
  \begin{math}
    \Paren{
      \begin{matrix}
        1 & \rho \\
        \rho & 1
      \end{matrix}}
  \end{math}
} and  $t\ge 0$ is
such that $\Prob[(x,y)\sim\cG_\rho]{x\ge t}=\mu$.

\begin{theorem}
\label{thm:noisestability} 

Let $\nu>0$, $\rho\in(0,1)$ and let $\Omega$ be a finite probability
space.
Then, there exists $\tau,\delta>0$ such that the following holds:
Every function $f\from \Omega^R\to [0,1]$ satisfies either
\begin{displaymath}
  \iprod{f,T_{\rho} f} 
  \le \Gamma_\rho(\E f) + \nu
  \mper
\end{displaymath}
or $\max_{i\in[R]} \Inf_i(T_{1-\delta} f) >\tau$.
(Here, $T_\rho$ and $T_{1-\delta}$ are the natural noise operators on
$L_2(\Omega^R)$ with correlation parameters $\rho$ and $1-\delta$ as
defined above.)
\end{theorem}

Below we define generalizations of the operators $M_x$ and $M$ used in
\pref{sec:hardn-balanc-separ}. 
We show that these operators can be viewed as somewhat extended
versions of the noise operators which randomize each coordinate of a product space with some
probability. 
The operators we define can be viewed as noise operators with additional ``leakage'' property, in
the sense that part of the output encodes the information about which coordinates were randomized. 
The second eigenvalue of these operators
can be easily estimated by relating it to the eigenvalue of the corresponding noise operator.

\paragraph{Random walks with leaked randomness}
\label{sec:random-walks-with}

\renewcommand{\super}[2]{#1^{\scriptscriptstyle #2}}

Suppose we have a collection of graphs $\set{G_z}_{z\in \cZ}$ with the
same vertex set $V$ (and with the same stationary distribution).
We consider two (reversible) random walks defined by this collection
and compare their spectral properties.
The first random walk is defined on $V$.
If the current state is $\super x1$, we choose the next state $\super
x2$ by sampling a random index $z\sim \cZ$ and then taking two
random steps from $\super x 1$ in $G_z$, i.e., we sample $x\sim
G_z(\super x 1)$ and $\super x2\sim G_z(x)$.
The second random walk is defined on $V\times \cZ$. 
If the current state is $(\super x 1,\super z 1)$, we choose the next
state $(\super x 2,\super z 2)$ by sampling a random neighbor $x$ of
$\super x 1$ in $G_{\super z 1}$, then we choose a random index
$\super z 2\sim \cZ$ and a random neighbor $\super x 2\sim
G_{\super z 2}(x)$ according to $G_{\super z 2}$.
The following lemma shows that these two random walks have the same
non-zero eigenvalues.
(Recall that we identify graphs with their stochastic operators.)

\begin{lemma}
  \label{lem:leaky-random-walk}
  Let $(\cZ,\nu)$ be a finite probability
  space and let $\set{G_z}_{z\in\cZ}$ be a family of graphs with
  the same vertex set $V$ and stationary measure $\mu$.
  Then the following two graphs have the same non-zero eigenvalues:
  \begin{itemize}
  \item the graph  $\E_{z\sim \cZ}G_z^2$ on $V$,
  \item the graph $H$ on $V\times \cZ$
    defined by
    \begin{displaymath}
      H f(\super x1,\super z1) 
      = \E_{x\sim G_{\super z1}(\super x1)} \E_{\super z2\sim \cZ}
      \E_{\super x2\sim G_{\super z2}(x)} f(\super x2,\super z2)
      \mper
    \end{displaymath}
  \end{itemize}
\end{lemma}
\begin{proof}
  Let $M$ be the following linear operator on $L_2(V\times \cZ)$,
  \begin{displaymath}
    M f(x,z) = \E_{z'\sim \cZ} \E_{x'\sim G_{z'}(x)} f(x',z')
    \mper
  \end{displaymath}
  Notice that its adjoint operator $M^*$ (with respect to the inner
  product in $L_2(V\times \cZ)$) is given by
  \begin{displaymath}
    M^* f(x,z) = \E_{x'\sim G_z(x)} \E_{z'\sim\cZ} f(x',z')
    \mper
  \end{displaymath}
  (The operator above is the adjoint of $M$, because each of the
  random walks $G_z$ are reversible and have the same stationary
  measure.).
  The graph $H$ corresponds to the operator $M^* M$, which has the
  same non-zero eigenvalues as $M M^*$.
  The operator $MM^*$ is given by
  \begin{displaymath}
    MM^* f(\super x1,\super z1) %
    = \E_{z\sim\cZ}\E_{x\sim G_{z}(\super x1)} %
    \E_{\super x2\sim G_{z}(x)} \E_{\super z2\sim\cZ}  %
    f(\super x2,\super z2)
    \mper
  \end{displaymath}
  The subspace
  \begin{math}
    \set{ f\mid \forall x\in V.~\E_{z} f(x,z)=0 }
    \sse L_2(V\times\cZ)
  \end{math}
  is part of the kernel of $MM^*$.
  Hence, all eigenfunctions with non-zero eigenvalue are in the
  orthogonal complement of this space.
  The orthogonal complement consists of all functions $f$ such that
  $f(x,z)$ does not depend on $z$.
  Let $f$ be such a function and set $f(x)=f(x,z)$.
  Then,
  \begin{displaymath}
    M M^* f(\super x 1) = \E_{z\sim\cZ} \E_{x\sim G(\super
      x1)}\E_{\super x2\sim G(x)} f(\super x2)
    = \E_{z\sim\cZ} G_z^2 f(\super x1)
    \mper
  \end{displaymath}
  Thus, $M M^*$ acts on this subspace in the same way as $\E_zG_z^2$,
  which means that the two operators have the same eigenfunctions (and
  eigenvalues) in this space.
\end{proof}

\paragraph{Noise graph with leaked randomness}
\label{sec:noise-graph-with}

\newcommand{\betabot}{\set{\bot,\top}^R_\beta}

Let $\betabot$ be the $\beta$-biased $R$-dimensional
boolean hypercube. 
If we sample a random point $z$ from this space, then $z_i=\top$ with
probability $\beta$, independently for each coordinate $i\in [R]$.

Let $(\Omega,\nu)$ be a finite probability space.
For $z\in \betabot$ and $x\in \Omega^R$, let $M_z(x)$ be the
distribution over $\Omega^R$ obtained by ``rerandomizing'' every
coordinate of $x$ where $z$ has value $\bot$.
In order to sample $x'\sim M_z(x)$, we sample $x'_i\sim \Omega$,
independently for every coordinate $i\in[R]$ with $z_i=\bot$.
If $z_i=\top$, then we copy the value of $x$ in this coordinate so
that $x'_i=x_i$.
Observe that $\E_{z\sim \betabot} M_z=T_{\beta,\Omega}^{\tensor R}$ is
the usual noise graph on $\Omega^R$ with correlation parameter
$\beta$, as defined previously in this section.

Consider the following stochastic linear operator $M$ on
$L_2(\Omega^R,\betabot)$,
\begin{equation}
  \label{eq:funny-noisegraph}
  M f(x,z) = \E_{z'\sim \betabot} \E_{x'\sim M_z(x)} f(x',z')
  \mper
\end{equation}

The following lemma shows that the second largest singular value of
$M$ is the same as the second largest eigenvalue of the corresponding
noise graph.
\begin{lemma}
\label{lem:leaky-noise-graph}
  Let $f\in L_2(\Omega^R,\betabot)$ 
  and let $M$ be as in \pref{eq:funny-noisegraph}. 
  Then, 
  \begin{displaymath}
    \snorm{ M f} 
    \le (\E f)^2 + \beta\cdot \Paren{\snorm{f}-(\E f)^2}
    \mper
  \end{displaymath}
\end{lemma}

\begin{proof}
  We have $\snorm{M f} = \iprod{f,M^* M f}$ where $M^*$ is the adjoint
  of $M$.
  This operator $M^* M$ is the same as the (second) operator in
  \pref{lem:leaky-random-walk} for $G_z=M_z$. 
  Hence, $M^* M$ has the same non-zero eigenvalues as $\E_z M_z^2$.
  From the definition of $M_z$, it is clear that $M_z^2=M_z$.
  Further, $T=\E_z M_z$ is the noise operator on $\Omega^R$ with
  correlation parameter $\beta$.
  We conclude that $M^* M$ has second largest eigenvalue $\beta$.
  The lemma follows from \pref{fact:second-eigenvalue}.
\end{proof}

\section{Reduction between Expansion Problems}
\label{sec:reduction}

Let $G$ be a graph with vertex set $V$ and stationary measure $\mu$.
Our reduction maps $G$ to a graph $H$ with vertex set $V^R \times
\Omega^R$ for $\Omega=[q]\times \betabias$. 
Here, $R,q\in \N$ and $\beta>0$ are parameters of the reduction.
We impose the natural product measure on $\Omega$,
$$ \Pr\inparen{(\alpha,z)} = \begin{cases} \frac{\beta}{q} & \text{ if
} z = \top \\
\frac{(1-\beta)}{q} & \text{ if
} z = \bot \end{cases} \qquad \forall \alpha \in [q]$$

As before, we describe $H$ in terms of a probabilistic process defined by $G$, which generates the
edge distribution of $H$.
(See \pref{fig:ssereduction} for a more condensed description.)
The process uses the following three auxiliary graphs (already
introduced in \Sref{sec:preliminaries} and \Sref{sec:additional-prelims}):
\begin{itemize}\item 
  First, the noise graph $T_V \seteq T_{1-\e_V,V}^{\tensor R}$, which
  resamples independently every coordinate of a given $R$-tuple $A\in
  V^R$ with probability $\e_V$.
  (Here, $\e_V>0$ is again a parameter of the reduction.
  We think of $\e_V$ as rather small compared to other parameters.)
This noise effectively adds a copy of the complete graph with weight $\e_V$ to $G$, which we assumed
in \Sref{sec:struct-uniq-games}.
 \item 
  Next, the noise graph $T_\Omega\seteq T_{\rho,\Omega}^{\tensor R}$,
  which resamples independently every coordinate of a given $R$-tuple
  $(x,z)\in\Omega^R$ with probability $1-\rho$.
  (For $x\in[q]^R$ and $z\in\betabias^R$, we write $(x,z)\in \Omega^R$
  to denote the tuple obtained by merging corresponding coordinates of
  $x$ and $z$ to an element of $\Omega$.
  In other words, we identify $[q]^R\times \betabias^R$ and
  $\Omega^R$.)
  The correlation parameter $\rho$ of $T_\Omega$ is the most important
  parameter of the reduction, because the graph $T_\Omega$ plays the
  role of a dictatorship test gadget in our reduction. We think of $\rho$ as being
  close to 1.
\item Finally, we consider the graph $M_z$ on $\tOmega^R$ for
  $\tOmega=V\times [q]$ and $z\in \betabias^R$.
  For $(A,x)\in \tOmega^R$, the graph $M_z$ resamples every coordinate
  in which $z$ has value $\bot$.
\end{itemize}
Our reduction proceeds in three phases:

In the first phase, we sample a random vertex $A\in V^R$ and take
two independent random steps from $A$ according to the graph
$T_VG^{\tensor R}$, i.e., we sample $\tB$ and $\tC$ from the
distribution $T_V G^{\tensor R}(A)$.
We end the first phase by sampling two random permutations $\pi_B$ and
$\pi_C$. The permutations are required to satisfy the property that if
we divide the domain $[R]$ into contiguous blocks of size $R/k$, then
each such block is permuted in place. We define the set $\Pi_k$ of such 
permutations as
\[ \Pi_k 
\seteq \inbrace{\pi \in S_R  ~\suchthat~ \forall j \in \{0, \ldots, k-1\}.~~
  \pi\inparen{\{\nfrac{jR}{k}+1,\ldots,\nfrac{(j+1)R}{k}\}} = 
\{\nfrac{jR}{k}+1,\ldots,\nfrac{(j+1)R}{k}\}} \mper\]
This phase exactly corresponds to the reduction from \SSE to
\uniquegames in \cite{RaghavendraS10}.
%

In the second phase, we sample a random $R$-tuple $(x_A,z_A)$ in
$\Omega^R$ and take two independent random steps from $(x_A,z_A)$
according to the graph $T_\Omega$, i.e., we sample $(x_B,z_B)$ and
$(x_C,z_C)$ from $T_\Omega(x_A,z_A)$.
This phase corresponds to typical dictatorship test reduction (as in
\cite{KhotKMO07}).

In the third phase, we apply the graphs $M_{z_B}$ and $M_{z_C}$ to the 
$R$-tuples $(\tB,x_B)$ and $(\tC,x_C)$ respectively, to obtain $(B',x'_B)$ and $(C',x'_C)$.
The final step of this phase is to output an edge between
$\pi_B(B',x'_B,z_B)$ and $\pi_C(C',x'_C,z_C)$.
(For a permutation $\pi$ of $[R]$ and an $R$-tuple $X$, we denote by
$\pi(X)$ the permutation of $X$ according to $\pi$, so that
$(\pi(X))_{\pi(i)}=X_i$.)

We remark that the random permutations $\pi_B$ and $\pi_C$ in the
first phase and the resampling according to $M_z$ in the third phase
introduce symmetries in the graph $H$ that effectively identify
vertices.
In particular, any two vertices in $V^R\times \Omega^R$ of the form
$(A,x,z)$ and $\pi(A,x,z)$ have the same neighbors in $H$ (i.e., the
distributions $H(A,x,z)$ and $H(\pi(a,x,z))$ are identical).
This kind of symmetry has been used in integrality gap constructions
(see \cite{KhotV05}) and hardness reductions (see
\cite{RaghavendraS10}).

The kind of symmetry introduced by the $M_z$ graph in the third phase
seems to be new. 
In the third phase, we effectively identify vertices $(A,x,z)$ and
$(A',x',z)$ if they differ only in the coordinates in which $z$ has
value $\bot$.
Formally, the vertex $(A,x,z)$ has the same distribution of neighbors
as the vertex $(A',x',z)$ if $(A',x')$ is sampled from $M_z(A,x)$.

\begin{figure}[!h]
  \begin{mybox}
\textbf{    The Reduction }
    \medskip

    {\sf Input}: A weighted graph $G$ with vertex set $V$. \\
    \textsf{Parameters}: $R,q, k \in \N$, and $\e_V, \beta, \rho > 0$.\\
    {\sf Output}: A graph $H=(V_H,E_H)$
    with vertex set $V_H = V^R\times [q]^R\times 
    \set{\top,\bot}_\beta^R$.

\medskip
Let $\Pi_k$ denote the set of permutations of $[R]$ which permute each
block of size $R/k$ in-place i.e.
\[ \Pi_k \seteq \inbrace{\pi \in S_R  ~\suchthat~ \forall j \in \{0, \ldots, k-1\}.~~
  \pi\inparen{\{\nfrac{jR}{k}+1,\ldots,\nfrac{(j+1)R}{k}\}} 
= \{\nfrac{jR}{k}+1,\ldots,\nfrac{(j+1)R}{k}\}} \mper\]

    \medskip The weight of an edge in $E_H$ is proportional to the
    probability that the following probabilistic process outputs this
    edge:

\begin{itemize}
\newcounter{reductionstep}
\item {\sf Reducing from \SSE to \uniquegames}.

  \begin{list}{\arabic{reductionstep}.}{\usecounter{reductionstep}}
  \item Sample an $R$-tuple of vertices $A \sim V^R$.
  \item Sample two random neighbors $B,C \sim G^{\tensor R}(A)$ of $A$.
  \item Sample $\tB \sim T_{V}(B)$ and $\tC \sim T_V(C)$.
  \item Sample two permutations $\pi_B,\pi_C \in \Pi_k$
 \end{list}

\item {\sf Combination with long code gadgets}.

\begin{list}{\arabic{reductionstep}.}{\usecounter{reductionstep}}
  \setcounter{reductionstep}{5}
\item Sample $(x_A,z_A) \in \Omega^R$, where $\Omega=[q]\times
  \set{\bot,\top}_\beta$.
\item Sample $(x_B, z_B),(x_C, z_C) \sim T_\Omega(x_A,z_A)$.
\end{list}

\item {\sf Redistributing the edge weights}

\begin{list}{\arabic{reductionstep}.}{\usecounter{reductionstep}}
  \setcounter{reductionstep}{6}
\item Sample $(B',x'_B)\sim M_{z_B}(\tB,x_B)$ and $(C',x'_C)\sim
  M_{z_C}(\tC,x_C)$
\item Output an edge between $\inparen{\pi_B(B',x'_B,z_B)}$ 
and $\inparen{\pi_C(C',x'_C,z_C)}$.
\end{list}
\end{itemize}
\end{mybox}
  \caption{Reduction between expansion problems}
  \label{fig:ssereduction}
\end{figure}

\begin{remark}[Reduction to \uniquegames with expansion]
\label{rem:expanding-ug-reduction}
  We note that the above reduction can also be viewed as creating a
  \uniquegames instance with alphabet size $q$.  For a vertex $(A,x,z) \in
  V_H$ and $l \in [q]$, let $(A,x,z) + l$ denote the vertex $(A,x',z)$,
  where $x'_i \equiv x_i + l \mod q$ for all $i \in [R]$. We define an
  equivalence relation on $V_H$ by taking $(A,x,z) \equiv (A,x,z) + l$ for
  all $A,x,z$ and $l \in [q]$. Let $H/[q]$ be a graph with one vertex for
  each equivalence class of the above relation.  Also, for each edge in
  $E_H$, we add an edge in $H/[q]$ between the equivalence classes
  containing the corresponding vertices of $V_H$. We claim that $H/[q]$ can
  then be viewed as a \uniquegames instance
  \footnote{In the terminology used in the literature, one can say that the
    graph $H$ is a \emph{label-extended} graph of a \uniquegames instance
    with alphabet size $[q]$.}
  as described in \pref{thm:expanding-unique-games}.

  We now describe the constraints for the edges in $H/[q]$. We identify
  each equivalence class with an arbitrarily chosen representative element
  in it. For $(A,x,z) \in V_H$, let $\rep{(A,x,z)}$ denote the
  representative of the equivalence class containing it. Consider an edge
  in $E_H$ between $\inparen{\pi_B\inparen{B',x_B',z_B}}$ and
  $\inparen{\pi_C(C',x_C',z_C)}$.  Let $\pi_B(B',x_B',z_B) =
  \rep{\pi_B(B',x_B',z_B)} + l_B$ and $\pi_C(C',x_C',z_C) =
  \rep{\pi_C(C',x_C',z_C)} + l_C$. Then the constraint corresponding to
  this edge requires that an assignment mapping vertices in $H/[q]$ to
  $[q]$ must satisfy
  \begin{equation*}
    F\inparen{\rep{\pi_B(B',x_B',z_B)}} + l_B 
    \equiv F\inparen{\rep{\pi_C(C',x_C',z_C)}} + l_C \mod q
    \mper
  \end{equation*}

  We note that the expansion properties of $H$ are inherited by $H/[q]$,
  since any set of measure $\mu$ in $H/[q]$ is also a set of measure $\mu$
  in $H$. In the \yes case, each of the sets $S_1,\ldots,S_q$ mentioned in
  \pref{thm:general-sse} will provide an assignment for the above
  \uniquegames instance, satisfying $1-\e-o(\e)$ fraction of the
  constraints. In the \no case, we will argue that each assignment
  corresponds to a set of measure $1/q$ in $H$, and the unsatisfiability of
  the instance will follow from the expansion of the corresponding sets in
  $H$.
\end{remark}
\noindent

\subsection{Completeness} \label{sec:completeness}

\begin{lemma}\label{lem:completeness}
Let $H = (V_H,E_H)$ be constructed from $G=(V,E)$ as in the reduction in 
\pref{fig:ssereduction}. If there is a set $S \subseteq V$ satisfying 
$\vol(S) \in \insquare{\frac{k}{10\beta R}, \frac{k}{\beta R}}$
and $\Phi_G(S) \leq \eta$, then there
exists a partition $S_1,\ldots,S_q$ of $V_H$ satisfying:
\begin{enumerate}
	\item \label{item:prop1} For all $(A,x,z) \in V_H$ and $l,l' \in [q]$, 
~$(A,x,z) \in S_l \implies (A,x,z) + l' \in S_{l+l'}$. 
\item \label{item:prop3} For each all $l \in [q]$, ~$\Phi_{H}(S_l) \leq  
2(1-\rho^2 + \eta+2\e_V) + (1-\rho^2 + \eta + 2\e_V)^2 + \frac{(1-\rho^2)\beta}{\rho^2} + 2^{-\Omega(k)}$\mper
\end{enumerate}
\end{lemma}

Note that the first property, together with the fact that $S_1,\ldots,S_q$ form a partition 
also implies that for all $l \in [q]$, $\vol(S_l) = \tfrac1q$.
\begin{proof}
We first describe a procedure for assigning vertices in $V_H$ to 
$S_1,\ldots, S_q$. This procedure assigns all but $2^{-\Omega(k)}$
fraction of the vertices, which we shall distribute arbitrarily later.
Let $(A,x,z)$ be a vertex in $V_H$, where
$A \in V^R$, $x \in [q]^R$ and 
$z \in \{\top,\bot\}^R$.

For all $j \in [k]$, we define the sets 
$W_j \seteq \inbrace{i \in \{\nfrac{(j-1)R}{k}+1, \ldots, \nfrac{jR}{k}\} \suchthat  z_i \neq \bot }$. 
Let $A(W_j)$ denote the multiset
$A(W_j) = \inbrace{A_i ~|~ i \in W_j}$. We take,
\[ j^* = \inf\inbrace{j \suchthat \card{A(W_j) \cap S} = 1}\]
If $\card{A(W_j) \cap S} \neq 1$ for any $j \in [k]$, then we do not
assign the vertex $(A,x,z)$ to any of the set $S_1, \ldots, S_q$.
Else, let $A_{i^*}$ be the unique element in 
$A(W_{j^*}) \cap S$. We assign  \[(A,x,z) \in S_{x_{i^*}}.\]
Note that the assignment to sets is determined only by the coordinates
$i \in [R]$ for which $z_i \neq \bot$. The first property is easily seen to be
satisfied for all the vertices that are assigned, as the sets 
$\inbrace{W_j}_{j \in [k]}$ are identical for 
$(A,x,z)$ and $(A,x,z)+l$, for any $l \in [q]$. 
The following claim proves that most vertices are indeed assigned to one of
the sets $S_1,\ldots, S_q$.

\begin{claim}\label{claim:comp-volume}
$\Prob[(A,x,z) \sim V_H]{\card{A(W_j) \cap S} \neq 1 ~\forall j \in [k]} 
~\leq~ 2^{-\Omega(k)}$.
\end{claim}
\begin{proof}
Note that over the choice of a random $(A,x,z) \in V_H$, the intersection sizes 
$\card{W_1 \cap S}, \ldots, \card{W_k \cap S}$ are independent random 
variables distributed as $\text{Binomial}(\mu \beta, \nfrac{R}{k})$.
The probability that all of them are not equal to 1, can then be bounded as
\begin{eqnarray*}
\Prob[(A,x,z) \sim V_H]{\card{A(W_j) \cap S} \neq 1 ~\forall j \in [k]}
&=& \inparen{1-\frac{R}{k} \cdot \mu \beta \cdot (1-\mu \beta)^{\nfrac{R}{k}-1}}^{k}\\
&\leq& \inparen{1-\frac{R}{k} \cdot \frac{k}{10R} \cdot (1-\nfrac{k}{R})^{\nfrac{R}{k}}}^{k}
~\leq~ \inparen{1-\frac{1}{30}}^k
\mper
\end{eqnarray*}
The last inequality assumes that $R/k \geq 4$ so that $(1-\nfrac{k}{R})^{\nfrac{R}{k}} > 1/3$. 
\end{proof}

We now bound the expansion of these sets. A random edge is between two tuples of the form
$\inparen{\pi_B(B', x_B',z_B)}$ and $\inparen{\pi_C(C',x_C',z_C)}$, where $\pi_B$ and $\pi_C$ 
are two random permutations in $\Pi_k$ and $B',C'$ 
are generated from $G^{\tensor R}$ as in \pref{fig:ssereduction}. For a fixed $l \in [q]$, 
the expansion of $S_l$ is equal to the following
probability taken over the choice of a random edge
\begin{equation*}
\Prob{\inparen{\pi_C(C',x_C',z_C)} \notin S_l \given 
\inparen{\pi_B(B', x_B',z_B)} \in S_l} \\
= \Prob{\inparen{C',x_C',z_C} \notin S_l \given \inparen{B', x_B',z_B} \in S_l}.
\end{equation*}
Here we used the fact that the membership in a set $S_l$ is invariant under permutations
in $\Pi_k$. The following claim analyzes the above probability.

\begin{claim}
$\Prob{\inparen{C',x_C',z_C} \notin S_l \given \inparen{B', x_B',z_B} \in S_l}
~\leq~ 2(1-\rho^2 + 2\epsilon_V + \eta) + (1-\rho^2 + 2\epsilon_V + \eta) ^2 + \frac{(1-\rho^2) \beta}{\rho^2}$.
\end{claim}
\begin{proof}
Let  $\inbrace{W_{j}^{(B)}}_{j \in [k]}$ denote the multisets 
$W_j^{(B)} = \inbrace{i \in \inbrace{\nfrac{(j-1)k}{R+1},\ldots,\nfrac{jk}{R}} \suchthat (z_B)_i
\neq \bot }$ and let $\inbrace{W_j^{(C)}}_{j \in [k]}$ be defined similarly. Define  
$j_B^* = \inf\inbrace{j \suchthat \card{B'(W_j^{(B)}) \cap S} = 1}$ 
when the set on the right is non-empty and $k+1$
otherwise. Let $j_C^*$ be the analogous quantity for $C'$. 
In the cases when $j^*_B, j_C^* \leq k$, let 
$W_{j_B^*}^{(B)} \cap S = \inbrace{B'_{i_B^*}}$ and  
$W_{j_C^*}^{(C)} \cap S = \inbrace{C'_{i_C^*}}$.
We can bound the required probability by the probability that either $j_B^* \neq j_C^*$
or $i_B^* \neq i_C^*$ or $(x_C')_{i_C^*} \neq l$.
\begin{align*}
\Prob{\inparen{C',x_C',z_C} \notin S_l \given \inparen{B', x_B',z_B} \in S_l} \leq~~
&\Prob{j_C^* \neq j_B^* \given \inparen{B', x_B',z_B} \in S_l}\\
+&\Prob{\inparen{j_C^* = j_B^*} \wedge \inparen{i_C^* \neq i_B^*} \given \inparen{B', x_B',z_B} \in
  S_l}\\
+ &\Prob{\inparen{j_C^* = j_B^*} \wedge \inparen{i_C^* = i_B^*} \wedge \inparen{(x_C')_{i_C^*} \neq l} \given \inparen{B', x_B',z_B} \in
  S_l} \\
\leq~~ & \Prob{j_C^* \neq j_B^* \given j_B^* \leq k}
~+~ \Prob{i_C^* \neq i_B^* \given \inparen{j_B^* = j_C^*} \wedge \inparen{j_B^* \leq k}}\\
+ &\Prob{(x_C')_{i_C^*} \neq l \given \inparen{(x_B')_{i_B^*} = l} \wedge \inparen{i_B^* = i_C^*}}
\end{align*}
In the second inequality above, we drop conditionings that are irrelevant and use 
$\Prob{A \wedge B} \leq \Prob{A \given B}$. We now analyze each of the above terms separately.

The first term can be further split as
\begin{equation*}
\Prob{j_C^* \neq j_B^* \given j_B^* \leq k} 
~=~ \Prob{j_C^* > j_B^* \given j_B^* \leq k} + \Prob{j_C^* < j_B^* \given j_B^* \leq k} \mper
\end{equation*}
To have $j_C^* > j_B^*$,  it must be the case that $\card{W^{(C)}_{j_B^*} \cap S} \neq 1$, while
we also have $\card{W^{(B)}_{j_B^*} \cap S} = 1$ by definition of $j_B^*$. If $i_B^* \notin
W^{(C)}_{j_B^*}$, this must be because $(z_{C})_{i_B^*} = \bot$ or $C_{i_B^*}' \notin S$. The former
happens with probability $1-\rho^2$ as we already have that $(z_{B})_{i_B^*} = \top$. The latter
even happens with probability at most $\eta + 2\e_V$ as it could be due to the edge 
$(B'_{i_B^*},C'_{i_B^*})$ going out of $S$ or one of the vertices being perturbed by
$T_V$. Combining, we get a bound of $(1-\rho^2 + \eta + 2\e_V)$ for the case when $i'$ such that
$C'_{i'} \in S$ and the the events above must happen for $W_{j_{B}^*}^{(B)}$ and $i'$, giving again
a bound of $(1-\rho^2 + \eta + 2\e_V)$. The term $\Prob{j_C^* < j_B^* \given j_B^* \leq k}$ can
be bound identically. We then get
\begin{equation*}
\Prob{j_C^* \neq j_B^* \given j_B^* \leq k}  
~=~ 2 (1-\rho^2 + \eta + 2\e_V) 
\mper
\end{equation*}

We now consider the term 
$\Prob{i_C^* \neq i_B^* \given \inparen{j_B^* = j_C^*} \wedge \inparen{j_B^* \leq k}}$.
For this to happen, the above events must occur for \emph{both} the pairs 
$\inparen{i_B^*, W^{(C)}_{j_C^*}}$ and $\inparen{i_C^*, W^{(B)}_{j_B^*}}$. This gives
\begin{equation*}
\Prob{i_C^* \neq i_B^* \given \inparen{j_B^* = j_C^*} \wedge \inparen{j_B^* \leq k}}
~\leq~ (1-\rho^2 + \eta + 2\e_V)^2
\end{equation*}

Finally, given $j_B^* = j_C^*$ and $i_B^* = i_C^*$, the probability that
$(x'_B)_{i_B^*} \neq (x'_C)_{i_C^*}$ is  at most $\tfrac{1-\rho^2}{\rho^2} \beta$.
This is true because for any $i$, $((x'_B)_i,(z_B)_i)$ and $((x'_C)_i,(z_C)_i)$
are same with probability $\rho^2$ and uniform in $\Omega$ with probability
$1-\rho^2$. Also, for any $i$,  $i = i^*_B = i^*_C$ in particular means that
$(z_B)_i = (z_C)_i = \top$. Conditioned on this, we can bound the probability
$(x'_B)_{i_B^*} \neq (x'_C)_{i_C^*}$ as
\begin{equation*}
\Prob{(x'_B)_{i} \neq (x'_C)_{i} \given (z_B)_i = (z_C)_i = \top}
~=~ \frac{(1-\rho^2) \cdot (1-\nfrac{1}{q}) \cdot \beta^2}{ (1-\rho^2) \cdot \beta^2 + \rho^2 \cdot
  \beta}
~\leq~ \frac{(1-\rho^2) \beta}{\rho^2}
\end{equation*}
Combining the bounds for the three terms proves the claim.
\end{proof}

It remains to partition the vertices not assigned to any of the sets $S_1,\ldots,S_q$. We simply
assign any such vertex $(A,x,z)$ to the set $S_{x_1}$. It is easy to see that $S_1,\ldots,S_q$ still
satisfy the first property. Since the measure of the extra vertices added to each set is
$\tfrac{2^{-\Omega(k)}}{q}$, the expansion of each set increases by at most $2^{-\Omega(k)}$.
\end{proof}

\subsection{Soundness} \label{sec:soundness}

Let $G$ be a graph with vertex set $V$ and stationary measure $\mu$.
Let $H$ be the graph obtained from the reduction in
\pref{fig:ssereduction}.
The vertex set of $H$ is $V^R\times \Omega^R$.
Recall that $\Omega=[q]\times\set{\bot,\top}_\beta$.
Let $f\from V^R \times\Omega^R\to [0,1]$. 
We think of $f$ as a cut in $H$ (or convex combination thereof).

We define two symmetrizations of $f$ as follows
\begin{displaymath}
  \bar f(A,x,z)
  = \E_{\pi \sim \Pi_k} f(\pi(A,x,z))
  \quad\text{and}\quad
   \bar f'(A,x,z)
  = \E_{(A',x')\sim M_z(A,x)} \bar f(\pi(A',x',z))
  \mper
\end{displaymath}
By the symmetries of the graph,
\begin{displaymath}
  \iprod{f,Hf} 
  = \iprod{\bar f,H \bar f} 
  = \iprod{\bar f',H\bar f'}
  \mper
\end{displaymath}
We write $\bar f'_A(x,z)=\bar f'(A,x,z)$ and consider the average (with noise)
of $\bar f'_B$ over the neighbors $B$ of a vertex $A$ in $G^R$,
\begin{displaymath}
  g_{A} = \E_{B\sim G^R(A)} \E_{\tB\sim T_V(B)} \bar f'_{\tB}
  \mper
\end{displaymath}
We will apply the techniques of \cite{KhotKMO07} to analyze the functions $g_A$.
We first express the fraction of edges that stay within the cut defined by $f$ in terms of the functions
$g_A$.
\begin{lemma}
  \label{lem:graph-form}
  \begin{displaymath}
    \iprod{f, H f}  = \E_{A\sim V^R}
    \snorm{T_\Omega g_{A} } \mper
  \end{displaymath}
\end{lemma}
\begin{proof}
  Using the construction of $H$ and the symmetry of $\bar f'$, we get
  \begin{align*}
    \iprod{f,Hf} &= \E_{A\sim V^R} \E_{(x,z)\sim \Omega^R }
    \Bigparen{\E_{B\sim G^R(A)}~ \E_{\tB\sim T_V(B)}~
      \E_{\substack{(x_B,z_B)\sim T_\Omega(x,z)}}~ \bar
      f'_{B}(x_B,z_B) }^2
    \\
    &=\E_{A\sim V^R} \E_{(x,z)\sim \Omega^R } \Bigparen{ %
      \E_{B\sim G^R(A)}~ \E_{\tB\sim T_V(B)}~ %
      T_\Omega \bar f'_{B}(x,z) }^2
    \\
    &=\E_{A\sim V^R} \E_{(x,z)\sim \Omega^R } \Bigparen{ T_\Omega
      g_A(x,z) }^2
    \\
    & = \E_{A\sim V^R} \snorm{T_\Omega g_{A} } \mper
    \qedhere    
  \end{align*}
\end{proof}

We now show that for most tuples $A$, the functions $g_A$ have the same expectation as $\E{f}$. To
this end, we show that $\E_A{\inparen{\E{g_A}}^2} \approx \inparen{\E{f}}^2$.
\begin{lemma}
\label{lem:variance-bound}
  \begin{displaymath}
    \E_{A\sim V^R}\Bigparen{\E_{\Omega^R} g_A }^2 
    \le     (\E f)^2 + \beta \snorm{f}
    \mper
  \end{displaymath}
\end{lemma}
\begin{proof}
  Let $\tOmega=V\times [q]$.  We have
  \begin{align}
    \E_{A\sim V^R}\Bigparen{\E_{\Omega^R} g_A }^2 & %
    = \E_{A\sim V^R} \Bigparen{ %
      \E_{B\sim G^R(A)} \E_{\tB\sim T_V(B)}~ \E_{(x,z)\sim\Omega^R}
      \bar f'(B,x,z)}^2
    \notag
    \\
    & \le \E_{A\sim V^R} \E_{B\sim G^R(A)} \E_{\tB\sim T_V(B)}~
    \E_{x}\Bigparen{ %
      \E_{z} \bar f'(B,x,z)}^2 \qquad \text{(Cauchy--Schwarz)}
    \notag
    \\
    & = \E_{(A,x)\sim \tOmega^R}\Bigparen{ %
      \E_{z} \bar f'(A,x,z)}^2 \mper
    \label{eq:variance-bound}
  \end{align}
  Let $M$ be the following stochastic operator on $L_2(\tOmega^R\times
  \betabot)$,
  \begin{displaymath}
    M h(A,x,z_0) = \E_{z\sim\betabot} \E_{\vbig (A',x')\sim M_z(A,x)}
    h(A',x',z)
    \mper
  \end{displaymath}
  Recall that
  \begin{math}
    \bar f'(A,x,z) = \E_{(A',x')\sim M_z(A,x)} \bar f(A',x',z).
  \end{math}
  With this notation, the right-hand side of \pref{eq:variance-bound}
  simplifies to $\snorm{M \bar f}$.
  Therefore, 
  \begin{displaymath}
    \E_{A\sim V^R}\Bigparen{\E_{\Omega^R} g_A }^2  
    \le \snorm{M\bar f}
    \le (\E \bar f)^2 + \beta \snorm{\bar f}
    \le (\E f )^2 + \beta \snorm{f}
    \mper
  \end{displaymath}
  The second inequality uses that $M ^*M$ has second largest
  eigenvalue $\beta$ (see \pref{lem:leaky-noise-graph}).
  The last inequality uses that $\bar f$ is obtained by applying a
  stochastic operator on $f$.  \qedhere
\end{proof}

The following lemma is an immediate consequence of the previous lemma
(\pref{lem:variance-bound}) and Chebyshev's inequality.

\begin{lemma}
  \label{lem:probability-bound}
  For every $\gamma>0$,
  \begin{displaymath}
    \Prob[A\sim V^R]{\vbig \E g_A \ge \E f + \gamma \sqrt{\E f} } 
    \le \nfrac \beta {\gamma^2} \cdot  \tfrac{\snorm{f} }{ \E f}
    \le \nfrac \beta{\gamma^2}
    \mper
  \end{displaymath}
\end{lemma}
\begin{proof}
  \pref{lem:variance-bound} shows that $\E_{A} (g_A - \E f)^2\le \beta
  \snorm{f}$.
  Hence,
  \begin{math}
    \Prob[A]{\abs{g_A-\E f}>\gamma \sqrt{\E f}} %
    \le \beta \snorm f / (\gamma^2 \E f).
  \end{math}
\end{proof}

\subsubsection{Decoding a \uniquegames assignment}
\label{sec:decod-uniq-assignm}

The goal is decode from $f$ an assignment $F\from V^R \to [R]$ that
maximizes the probability
\begin{equation}
  \label{eq:ug-successprobability}
  \E_{A\sim V^R, B\sim G^R(A)}
  \E_{\tA\sim T_V(A)}\E_{\tB\sim T_V(B)}
  \Prob[\pi_A,\pi_B \sim \Pi_k]{
    \inv{\pi_A}\paren{F\paren{\pi(A)}}
    =\inv{\pi_B}\paren{F\paren{\pi_B(B)}} 
  }\mper
\end{equation}
(The expression above is roughly the success probability of the
assignment $F$ for the \uniquegames instance obtained by applying the
reduction from \cite{RaghavendraS10} on $G$.)

As usual, we decode according to influential coordinates of $f$ (after
symmetrization).
More precisely, we generate a assignment $F$ by the following
probabilistic process:
For every $A\in V^R$, with probability $\half$, choose a random
coordinate in $\set{ i\in [R]\mid \Inf_i (T_{1-\delta} g_A)>\tau}$ and
with probability $\half$, choose a random coordinate in $\set{ i\in
  [R]\mid \Inf_i (T_{1-\delta} \bar f'_A)>\tau}$.
If the sets of influential coordinates are empty, we choose a
uniformly random coordinate in $[R]$.

\Mnote{Replacing $\delta$ in $T_{1-\delta}$ below by $\e$, and $\e$ in 
the error by $\nu$.}
The following lemma follows immediately from the techniques in
\cite{KhotKMO07}.
The reason is that \pref{eq:ug-successprobability} is the success
probability of the assignment $F$ for a \uniquegames instance defined
on $V^R$.
For $A\in V^R$, the function $g_A$ is the average over bounded
functions $f'_B\from \Omega^R\to[0,1]$, where $B$ is a random neighbor
of $A$ in the \uniquegames instance and where input coordinates of
$f'_B$ are permuted according to the constraint between $A$ and $B$.
More precisely, 
\begin{displaymath}
  g_A(x,z) = \E_{B\sim G^R(A)}\E_{\tB\sim T_V(B)} \E_{\pi_B} 
  f'_{\pi_B(\tB)}(\pi_B(x,z))
  \qquad\text{where}~f'_B(x,z) = \E_{(B',x')\sim M_z(B,x)} f(B',x',z)
  \mper
\end{displaymath}

\begin{lemma}
  \label{lem:influence-decoding}
  For every $\tau,\delta>0$, there exists a constant $c>0$ such that
  \begin{displaymath}
    \E_{F} \E_{A\sim V^R, B\sim G^R(A)}
    \E_{\tA\sim T_V(A)}\E_{\tB\sim T_V(B)}
    \Prob[\pi_A,\pi_B \sim Pi_k]{
    \inv{\pi_A}\paren{F\paren{\pi_A(\tA)}}
    =\inv{\pi_B}\paren{F\paren{\pi_B(\tB)}} 
  } 
  > c \Prob[A\sim V^R]{\exists i.~\Inf_i(T_{1-\delta}g_A)>\tau}
  \mper
  \end{displaymath}
\end{lemma}

\Dcomment{\begin{proof}
  Using following facts: 
  \begin{itemize}\item 
    $g_A = \E_{B\sim G^R(A)}~\E_{\tB\sim T_V(B)}~ \bar f'_A$,
  \item $\bar f'_A(\pi(x,y,z)) = \bar f'_A(x,y,z)$ for all
    permutations $\pi\from[R]\to [R]$,
  \item convexity of influences,
  \item a function $T_{1-\delta} f$ can have at most $\snorm
    f/\poly(\tau\delta)$ coordinates with influence at least $\tau$.
  \end{itemize}
\end{proof}
}

\newcommand{\trho}{{\tilde\rho}}

\begin{lemma}
  \label{lem:expansion-vs-influence}
  For every $\nu, \beta, \gamma > 0$, $q\in\N$, and $\rho\in(0,1)$, there exist
  $\tau,\delta > 0$ such that
  \begin{displaymath}   
    \iprod{f,Hf} < \Gamma_{\rho^2}\bigparen{\E f} + 2\gamma + \nu
    +\nfrac \beta {\gamma^2} 
    + \Prob[A\sim V^R]{\exists i.~\Inf_i(T_{1-\delta}g_A)>\tau}
    \mper
 \end{displaymath}%
\end{lemma}
\begin{proof}
  Recall that \pref{lem:graph-form} shows
  \begin{math}
    \iprod{f, H f} = \E_A \snorm{T_\Omega g_A} 
    \mper
  \end{math}
  The operator $T_\Omega$ is an $R$-fold tensor operator with second
  largest eigenvalue $\rho$. 
  The invariance principle (\pref{thm:noisestability}) asserts that there
  exist $\tau,\delta > 0$ such that $\snorm{T_\Omega g_A}\le
  \Gamma_{\rho^2}(\E g_A) +\nu$ if $\Inf_i(T_{1-\delta} g_A)\le \tau$
  for all coordinates $i\in[R]$.
  Together with \pref{lem:probability-bound}, we get
  \begin{align*}
    \iprod{f, H f} = \E_A \snorm{T_\Omega g_A} 
    &~\le  \Gamma_{\rho^2}\inparen{\E f + \gamma\sqrt{\E f}} + \nu + \nfrac\beta {\gamma^2}
    + \Prob[A\sim V^R]{\exists  i.~\Inf_i(T_{1-\delta}g_A)>\tau}
    \mper    \\
&~\le  \Gamma_{\rho^2}(\E f) + 2\gamma + \nu + \nfrac\beta {\gamma^{2}} 
    + \Prob[A\sim V^R]{\exists  i.~\Inf_i(T_{1-\delta}g_A)>\tau}
  \end{align*}
The second inequality above used that $\Gamma_{\rho^2}(\cdot)$ is 2-Lipschitz.
\end{proof}

Putting together \pref{lem:influence-decoding} and
\pref{lem:expansion-vs-influence}, we get the following lemma as
immediate corollary.

\begin{lemma}
  \label{lem:uniquegame-decoding}
  For every $\beta>0$, $q\in\N$, and $\rho\in(0,1)$, there
  exists $\zeta>0$ such that either
  \begin{displaymath}
    \iprod{f,Hf} 
    < \Gamma_{\rho^2}\bigparen{\E f} 
    +5\beta^{1/3} \mcom
  \end{displaymath}
  or there exists an assignment $F\from V^R\to [R]$ such that the
  probability \pref{eq:ug-successprobability} is at least $\zeta$.
\end{lemma}
\begin{proof}
Choosing $\gamma = \nu = \beta^{1/3}$ in \pref{lem:expansion-vs-influence},
we get that for some $\tau,\delta > 0$,
\begin{align*}
    \iprod{f, H f} = \E_A \snorm{T_\Omega g_A} 
&~\le  \Gamma_{\rho^2}(\E f) + 4\beta^{1/3} + 
\Prob[A\sim V^R]{\exists  i.~\Inf_i(T_{1-\delta}g_A)>\tau} \mper
\end{align*}
Taking $\zeta = c\beta^{1/3}$ for the constant $c$ in 
Lemma \pref{lem:influence-decoding} then proves the claim.
\end{proof}

\subsubsection{Decoding a very small non-expanding set in \texorpdfstring{$G$}{G}}

The following lemma is a slight adaptation of a result in
\cite{RaghavendraS10} (a reduction from \SSE to \uniquegames).
We present a sketch of the proof in \pref{sec:sse-to-ug}.

\begin{lemma}
  \torestate{
 \label{lem:smallset-soundness}
Let $G$ be graph with vertex set $V$. 
Let a distribution on pairs of tuples $(\tA,\tB)$ be defined by choosing
$A \sim V^R, B \sim G^{\tensor R}(A)$ and then $\tA \sim T_V(A), \tB \sim T_V(B)$. Let $F \from V^R \to [R]$
be a function such that over the choice of random tuples and two random permutations 
$\pi_A,\pi_B \in \Pi_k$,
\[ \Pr_{(\tA,\tB)} \Prob[\pi_A,\pi_B \sim \Pi_k]{\inv{\pi_A}(F(\pi_A(\tA))) = \inv{\pi_B}(F(\pi_B(\tB)))}
~~\geq~~ \zeta.\]
Then there exists a set $S \subseteq V$ with $\vol(S) \in \insquare{\frac{\zeta}{16R}, \frac{3k}{\e_V R}}$ satisfying
$\Phi_G(S) \leq 1- \tfrac{\zeta}{16k}$.
}
\end{lemma}

Putting together \pref{lem:smallset-soundness} and
\pref{lem:uniquegame-decoding}, we get the following lemma --- the
main lemma for the soundness of the reduction.

\begin{lemma}
  \label{lem:soundness}
  Let $G$ be a graph with vertex set $V$.
  Let $H$ be the reduction in \pref{fig:ssereduction} applied to $G$
  with parameters $R,q\in\N$, $\e_V,\beta>0$ and $\rho\in(0,1)$.
  The vertex set of $H$ is $V^R\times \Omega^R$, where
  $\Omega=[q]\times \set{\bot,\top}_\beta$.
  Then there exists
  $\zeta=\zeta(\beta,q,\rho)>0$ such that either 
  \begin{displaymath}
    \forall f\from V^R\times \Omega^R\to [0,1].~\forall \gamma>0.\quad
    \iprod{f,Hf} 
    < \Gamma_{\rho^2}\bigparen{\E f} 
    +5\beta^{1/3} \mcom
  \end{displaymath}
  or there exists a vertex set $S\sse V$ with $\mu(S)\in [\tfrac \zeta
  {R}, \tfrac {3k}{\e_V R}]$ and $\cond_G(S)\le 1-\zeta/k$.
\end{lemma}

\subsection{Putting things together}
\label{sec:putt-things-together}

\restatetheorem{thm:general-sse}

\begin{proof}
The follows by proper choice of parameters for the reduction in
\pref{fig:ssereduction}. Given $q,\e,\gamma$, we choose the various 
parameters in the reduction as below:

\begin{itemize}
\item $\rho = \sqrt{1-\tfrac\e2}$.
\item $\beta = \min\inbrace{\tfrac{\gamma^3}{200}, \e}$, so that the error 
$5\beta^{1/3} < \gamma$ in \pref{lem:soundness} and the error 
$\tfrac{1-\rho^2}{\rho^2} \beta = O(\e^2)$ in \pref{lem:completeness}.
\item $k = \Omega(\log(\nfrac 1 \e))$, so that the $2^{-\Omega(k)}$ error
term in \pref{lem:completeness} is $O(\e^2)$.
\item $\e_V = \e^2$ and $\eta = \min\inbrace{\e^2,\tfrac{\zeta}{k}}$.
Here, $\zeta = \zeta(\beta,q,\rho)$ is the constant given by \pref{lem:soundness}.
The above choices ensure that the error term $\e_V + \eta$ in \pref{lem:completeness}
are $O(\e^2)$ and $\eta \leq \tfrac{\zeta}{k}$ for applying \pref{lem:soundness}.
\item $M = \max \inbrace{\tfrac{k}{\beta\zeta}, \tfrac{3\beta}{\e_V}}$.
\item $R = \tfrac{k}{\beta \delta}$, where $\delta \in (0,\eta)$ is the one for which we
intend to show a reduction from \smallsetexpansion($\eta,\delta,M$).
\end{itemize}

Given and instance $G=(V,E)$ of \smallsetexpansion($\eta,\delta,M$),
let $H$ be the graph obtained from the reduction in
\pref{fig:ssereduction} with these parameters.
  
  From \pref{lem:completeness}, we get that the {\yes} case of
  \smallsetexpansion($\eta,\delta,M$) implies the {\yes} case of the above
  problem. On the other hand \pref{lem:soundness} gives that a
  contradiction to the {\no} case of the above problem produces a
  set $S$ in $G$ with measure between $\tfrac{\zeta}{R}$ and
  $\tfrac{3k}{\e_V R}$ with $\cond_G(S) \leq 1-\tfrac{\zeta}{k}$.
  By our choice of parameters, this is a set of measure between $\tfrac{\delta}{M}$ and $M\delta$
with expansion $1-\tfrac{\zeta}{k} \leq 1-\eta$. This contradicts the {\no} case of
  \smallsetexpansion($\eta,\delta,M$).
\end{proof}

\restatetheorem{thm:expanding-unique-games}
\begin{proof}
Let all the parameters for the reduction in \pref{fig:ssereduction} be chosen as in 
the proof for \pref{thm:general-sse}, replacing $\e$ by $2\e$. Let $H$ be the
graph generated by the reduction starting from an instance $G$ of \smallsetexpansion($\eta,\delta,M$).
Let $\uginst$ be the \uniquegames instance defined on the graph $H/[q]$, as described 
in \pref{rem:expanding-ug-reduction}.

We claim that any partition $S_1,\ldots,S_q$ of the vertices in $H$, satisfying the first property
in \pref{lem:completeness}, corresponds to an assignment to the vertices in $H/[q]$ and vice-versa.
A partition is simply a function $F \from V_H \to [q]$. Restricting the function to the
representatives of each equivalence class gives an assignment for the vertices in $H/[q]$.
Note that here $F$ also satisfies that 
$(A,x,z) = \rep{(A,x,z)} + l \implies F((A,x,z)) = F\inparen{\rep{(A,x,z)}} + l$.
Similarly, given an assignment $F$, we can extend it to all the vertices in $H$ by \emph{defining}
$F((A,x,z)) = F\inparen{\rep{(A,x,z)}} + l$ if $(A,x,z) = \rep{(A,x,z)} + l$.

In the \yes case, we construct an assignment to the \uniquegames instance $\uginst$ using the
partition $S_1,\ldots, S_q$. The fraction of edges  $\inparen{\pi_B(B',x_B',z_B),
  \pi_C(C',x_C',z_C)}$ that are not satisfied is exactly the probability that
$F\inparen{\pi_B(B',x_B',z_B)} \neq F\inparen{\pi_C(C',x_C',z_C)}$ for a random edge.
However, this is exactly $\E_{l \in [q]}{\Phi_H(S_l)}$ which is at most $2\e+o(\e)$ by
\pref{thm:general-sse}.

In the \no case, we note that we can construct a partition $S_1,\ldots,S_q$ from any assignment
$F$. The fraction of unsatisfied edges is again $\E_{l \in [q]} \Phi_H(S_l) \geq 1 -
q\inparen{\Gamma_{1-\e}(\nfrac 1 q) + \gamma}$ by \pref{thm:general-sse}. Also, any set $S$
in $H/[q]$ corresponds to a set $\tilde{S}$ in $H$ with $\mu(\tilde{S}) = \mu(S)$, where
$\tilde{S}$ contains all the vertices for each class in $S$. The edges leaving $S$ and $\tilde{S}$
are the same and hence their expansion is identical.
\end{proof}

\subsection*{Acknowledgments}

We are grateful to Subhash Khot for suggesting
that our techniques should also show that \uniquegames is \ssehard on
graph with high (small-set) expansion
(\pref{thm:expanding-unique-games}).
We also thank Oded Regev for insightful discussions and Boaz Barak 
for helpful comments on this manuscript.

\phantomsection
\addcontentsline{toc}{section}{References}
\bibliographystyle{amsalpha}
\bibliography{ssereductions}

\appendix

\section{Further Proofs}
\subsection{Reduction from \SSE to \uniquegames}
\label{sec:sse-to-ug}

In this section, we sketch a proof of the following slight adaption of
a result in \cite{RaghavendraS10}.

\restatelemma{lem:smallset-soundness}

\Mnote{Creating a small claim below for the reduction from partial to total unique games.}
Let $R' = R/k$ and let $\tA_{R'},\tB_{R'}$ denote tuples of length $R'$ generated by a
a process similar to the used for generating $\tA,\tB$ (which have length $R$).
Using the reduction from partial to total unique games in
\cite{RaghavendraS10}, we can show the following for completely random permutations 
(instead of block-wise random) permutations $\pi_A', \pi_B' \from [R'] \to [R']$.

\begin{claim}\label{clm:total-to-partial}
Given a function $F \from V^R \to [R]$ as above, there exists a function 
$F' \from V^{R'} \to [R']$ such that
\begin{displaymath}
   \Pr_{(\tA_{R'},\tB_{R'})} \Prob[\pi_A',\pi_B' \sim S_{R'}]{\inv{\pi_A'}(F'(\pi_A'(\tA_{R'}))) = \inv{\pi_B'}(F'(\pi_B'(\tB_{R'})))}
~~\geq~~ \zeta/k.
\end{displaymath}
\end{claim}
\begin{proof}
We construct a randomized function $F'$ which given an $R'$-tuple,
embeds it as one of the blocks (of size $R'$) in a random $R$-tuple,
and then outputs a value according to the value of $F$ on the $R$-tuple.

Formally, let $\tA_{R-R'}, \tB_{R-R'}$ denote tuples of size $R-R'$ generated
by independently picking each pair of coordinates to be an edge in $G$ with 
noise $\e_V$. For $j \in [k]$, let $\tA_{R-R'} +_j \tA_{R'}$ denote an $R$-tuple
generated by inserting $\tA_{R'}$ after the $(j-1)^{th}$ block in $\tA_{R-R'}$.
Let $\pi_A^{(j)}$ be a random permutation in $\Pi_k$ which is equal to
identity on the $j^{th}$ block. To define $F'$, we then generate a random
$j$, $\tA_{R-R'}, \tB_{R-R'}$ and fix it {\em globally}. For each input
$\tA_{R'}$, we then {\em independently} choose $T_A \in \{\tA_{R-R'}, \tB_{R-R'}\}$,
a permutation $\pi_A^{(j)}$, and define $F'(\tA_{R'})$ as:
\[
F'(\tA_{R'}) = \left\{\begin{array}{cl}
\inv{(\pi_A^{(j)})}\inparen{F\inparen{\pi_A^{(j)}(T_A) +_j \tA_{R'}}}- \nfrac{(j-1)R}{k} & 
\quad \inv{(\pi_A^{(j)})}\inparen{F\inparen{\pi_A^{(j)}(T_A) +_j \tA_{R'}}} \in (\nfrac{(j-1)R}{k}, \nfrac{jR}{k}] \\
~\\
1 & \quad\text{otherwise}
\end{array}\right.
\]
Let $\pi_A$ be the permutation which is $\pi_A'$ on the $j^{th}$ block and $\pi_A^{(j)}$ elsewhere.
Define $\pi_B$ similarly. Note that both $\pi_A,\pi_B$ are distributed as random elements of $\Pi_k$.
Conditioned on $T_A = \tA_{R-R'}, T_B = \tB_{R-R'}$ (or vice-versa), the required probability is at least
\begin{align*}
& \Pr_j \Pr_{(\tA_{R-R'},\tB_{R-R'})} \Pr_{(\tA_{R'},\tB_{R'})} 
\Prob[\pi_A,\pi_B \sim \Pi_k]{\inv{\pi_A}(F(\pi_A(\tA_{R-R'} +_j \tA_{R'}))) = 
\inv{\pi_B}(F(\pi_B(\tB_{R-R'} +_j \tB_{R'})))} \\
=~~ & \frac{1}{k} \cdot \Pr_{(\tA,\tB)} 
\Prob[\pi_A,\pi_B \sim \Pi_k]{\inv{\pi_A}(F(\pi_A(\tA))) = \inv{\pi_B}(F(\pi_B(\tB)))} 
~=~ \frac{\zeta}{k}\mper
\end{align*}
Since we have $T_A = \tA_{R-R'}, T_B = \tB_{R-R'}$ or vice-versa with probability 1/2, 
the required probability is at least $\zeta/2k$.
\end{proof}

Let $\zeta'$ denote $\zeta/2k$.

To construct the set $S\sse V$, we proceed as in \cite{RaghavendraS10} by defining
the influence of a 
single vertex on the output of $F$. For $U \in V^{{R'}-1}$ and $v \in V$, let $\tU \sim T_V(U)$
and $\tv \sim T_V(v)$. For $i \in {R'}$, we use $\tU +_i \tv$ to denote the tuple 
$(\tU_1, \ldots, \tU_{i-1}, v, \tU_i, \ldots, \tU_{R'}) \in V^{{R'}}$ in which $v$ is inserted at the
$i$th position. We define the function $F_U(v)$, which measures how often is the index of $v$
chosen by $F$, when applied to a random permutation $\pi$ of $\tU +_i \tv$.

\[ F_U(v) ~\seteq~ \E_{\tU \sim T_V(U)} \E_{\tv \sim T_V(v)} 
\Pr_{i \in [{R'}]} \Prob[\pi \in S_{R'}]{F(\pi(\tU +_i \tv)) = \pi(i)} \]

We shall need the following (slight variants of) statements proved in \cite{RaghavendraS10}. 
We include the proofs in the appendix for completeness.

\Mnote{Need to include proofs of these with modified notation.}
\begin{lemma}[\glorifiedmarkov]
  \label{lem:glorifiedmarkov}
  Let $\Omega$ be a probability space and let $X,Y\from \Omega\to
  \R_+$ be two jointly distributed non-negative random variables over
  $\Omega\,$.
  Suppose $\E X \le \gamma\E Y$.
  Then, there exists $\omega\in \Omega$ such that $X(\omega)\le
  2\gamma Y(\omega)$ and $Y(\omega)\ge \E Y/2$.
\end{lemma}

\begin{proposition}
\label{prop:properties}
Let $F \from V^{R'} \to [{R'}]$ satisfy
$\Pr_{(\tA,\tB)} \Prob[\pi_A,\pi_B]{\inv{\pi_A}(F(\pi_A(\tA))) = \inv{\pi_B}(F(\pi_B(\tB)))}\geq \zeta'$,
and the functions $F_U \from V \to [0,1]$ be defined as above. Then,
  \begin{enumerate}
  \item 
    \begin{equation}
      \E_{U \sim V^{{R'}-1}} \E_{v \sim V} F_U(v)
      = \tfrac1{R'}
      \mcom
      \label{eq:soundness-property-1}
    \end{equation}
  \item for all $U\in E^{{R'}-1}$,
    \begin{equation}
      \E_{v\sim V} F_U(v)
      \le \tfrac{2} {\e_V {R'}}
      \mcom
      \label{eq:soundness-property-2}
    \end{equation}
  \item 
    \begin{equation}
      \E_{(U,W) \in E^{{R'}-1}} %
      ~\Ex[(v_1,v_2) \in E]{F_U(v_1)F_W(v_2)} ~\geq~ \tfrac{\zeta'}{{R'}}
      \mper
      \label{eq:soundness-property-3}
    \end{equation}
  \end{enumerate}
\end{proposition}

Assuming Lemma \ref{lem:glorifiedmarkov} and Proposition \ref{prop:properties}, 
we can now complete the proof of Lemma \ref{lem:smallset-soundness}.

\begin{proof}
By (\ref{eq:soundness-property-3}) and (\ref{eq:soundness-property-1}), we have that
\begin{equation*}
\E_{(U,W) \in E^{{R'}-1}} \Ex[(v_1,v_2) \in E]{F_U(v_1)F_W(v_2)} ~~\geq~~ \tfrac{\zeta'}{{R'}}
~~=~~ \tfrac{\zeta'}{2} \cdot \E_{(U,W) \in E^{{R'}-1}} \Ex[(v_1,v_2) \in E]{F_U(v_1) + F_W(v_2)}
\end{equation*}
Using Lemma \ref{lem:glorifiedmarkov}, this gives that there exist $(U^*,W^*) \in E^{{R'}-1}$ such that
\newcommand{\fu}{F_{U^*}}
\newcommand{\fw}{F_{W^*}}
\begin{equation*}
\Ex[(v_1,v_2) \in E]{\fu(v_1)\fw(v_2)} \geq \tfrac{\zeta'}{2{R'}} \qquad \text{and}
\qquad \Ex[(v_1,v_2) \in E]{\fu(v_1)\fw(v_2)} \geq \tfrac{\zeta'}{4} \cdot \inparen{\E \fu + \E \fw}
\end{equation*}
We now construct the set $S$ randomly, by choosing each $v \in V$ to be in
$S$ with probability $(\fu(v)+\fw(v))/2$. We first check that the expected
volume of the set is large.
\begin{eqnarray*}
\E \vol(S) ~=~ \Ex[v \sim V]{\frac{\fu(v)+\fw(v)}{2}}
&=& \Ex[(v_1,v_2) \in E]{\frac{\fu(v_1)+\fw(v_2)}{2}}\\
&\geq& \Ex[(v_1,v_2) \in E]{\frac{\fu(v_1) \fw(v_2)}{2}} 
\qquad(\text{Using}~ a+b \geq ab ~\text{for}~a,b \in [0,1])\\
&\geq& \tfrac{\zeta'}{4{R'}}
\end{eqnarray*}
Combining this with (\ref{eq:soundness-property-2}), we get that 
$\E\vol(S) \in \insquare{\frac{\zeta'}{4{R'}},\frac{2}{\e_V {R'}}}$.
Also, by a Chernoff bound, we have that with probability
$1-\exp(-\Omega(|V|))$, $\vol(S) \in \insquare{\frac{\zeta'}{8{R'}},\frac{3}{\e_V {R'}}}$.

To show that the expansion of the set is bounded away from 1, we show a lower bound
on the expected number of edges that stay within the set, denoted by $G(S,S)$.
\begin{eqnarray*}
\E G(S,S) &=& 
\Ex[(v_1,v_2) \in E]{\inparen{\frac{\fu(v_1)+\fw(v_1)}{2}}\inparen{\frac{\fu(v_2)+\fw(v_2)}{2}}}\\
&=& \tfrac{1}{2} \cdot \E_{(v_1,v_2) \in E}{{\fu(v_1) \fw(v_2)}} ~+~ 
\tfrac{1}{4} \cdot \Ex[(v_1,v_2) \in E]{{\fu^2(v_1) + \fw^2(v_2)}}\\
&\geq& \tfrac{1}{2} \cdot \E_{(v_1,v_2) \in E}{{\fu(v_1) \fw(v_2)}}\\
&\geq& \tfrac{\zeta'}{4} \Ex[v \sim V]{(\fu + \fw)/2} ~~=~~ \tfrac{\zeta'}{4} \cdot \E\vol(S) 
\end{eqnarray*}
Thus, we have
\begin{equation*}
\E\insquare{G(S,S) - \tfrac{\zeta'}{8} \vol(S)} ~\geq~ \tfrac{\zeta'}{8}\E\vol(S) ~\geq~ \tfrac{\zeta'^2}{8{R'}}. 
\end{equation*}
In particular, we get that with probability at least $\tfrac{\zeta'^2}{16{R'}}$ over the choice of $S$,
$G(S,S) \geq \tfrac{\zeta'}{8} \cdot \vol(S)$. Hence, with probability 
$\tfrac{\zeta'^2}{16{R'}} - e^{-\Omega(|V|)}$, we have 
$\vol(S) \in \insquare{\frac{\zeta'}{8{R'}},\frac{3}{\e_V {R'}}}$ and
$G(S,S) \geq \tfrac{\zeta'}{8} \cdot \vol(S)$. For such a set we have
$\Phi_G(S) = 1 - (G(S,S)/\vol(S)) \leq 1 - \tfrac{\zeta'}{8}$, which proves the claim. 
\end{proof}

\subsection{Stronger Small-Set Expansion Hypothesis}
\label{sec:stronger-small-set}

\restateprop{prop:stronger-sse}

\begin{proof}
Let $\eta' = \tfrac{\eta}{M}$. The reduction is in fact, the trivial one which, 
given an instance $G=(V,E)$ of \smallsetexpansion($\eta',\delta$) 
treats as an instance of \smallsetexpansion($\eta,\delta,M$). 
If we are in the {\yes} case of \smallsetexpansion($\eta',\delta$), then there is a set $S$
with $\vol(S) = \delta$ and $\cond_G(S) \leq \eta' \leq \eta$. Hence, we are also in the {\yes} case of 
\smallsetexpansion($\eta,\delta,M$).

For the other direction, assume that we are {\em not} in the {\no} case of \smallsetexpansion($\eta,\delta,M$) and 
there exists a set $S$ with $\vol(S) \in \inparen{\tfrac{\delta}{M},\tfrac{\delta}{M}}$ and $\cond_G(S) \leq 1-\eta$.
Then the fraction of edges $G(S,S)$ stay inside $S$ is at least $\eta \cdot \vol(S)$. If $\vol(S) \geq \delta$, then we
randomly sample a subset $S'$ of $S$ with volume $\delta$. For each edge $(u,v) \sse S$, the chance that
$(u,v) \in S'$ is $\delta^2/\vol(S)^2$. Then
\[ \E\cond_G(S') = 1-\frac{\E G(S',S')}{\delta} \leq 1- \frac{(\delta^2/\vol(S)^2)\cdot \eta \cdot \vol(S)}{\delta} \leq 1-\tfrac{\eta}{M}.\]
Then, we cannot be in the {\no} case of \smallsetexpansion($\eta',\delta$). When $\vol(S) \leq \delta$, we simply create a set $S'$
by adding extra vertices to $S$ to increase its measure to $\delta$. Then,
\[ \cond_G(S') = 1-\frac{G(S',S')}{\delta} \leq 1- \frac{G(S,S)}{\delta} \leq 1 - \frac{\eta \cdot \vol(S)}{\delta} \leq 1- \tfrac{\eta}{M}.\]

\end{proof}

\subsection{Hardness of \mla and \BalancedSeparator} \label{app:mlabalanced}

\begin{corollary}[Hardness of \BalancedSeparator and \minbisection]
There is a constant $c$ such that for arbitrarily small $\epsilon > 0$, it is 
\ssehard to distinguish the following two cases 
for a given graph $G=(V,E)$:
\begin{description}
 \item[\yes:] There exists a cut $(S,V \setminus S)$ in $G$ such that 
$\vol(S) = \tfrac12$ and $\cond_G(S) \leq \epsilon + o(\epsilon)$.
 \item[\no: ] Every cut $(S,V\setminus S)$ in $G$, with 
$\vol(S) \in \inparen{\tfrac{1}{10},\tfrac12}$ satisfies 
$\vol_G(S) \geq c \sqrt{\epsilon} $.
\end{description}
\end{corollary}

\begin{proof}
The result follows immediately by applying \pref{thm:general-sse} with the given $\e$
and taking $q=2$,  $\gamma = o(\sqrt{\e})$. In the {\no} case we get that for all sets $S$ with
$\mu(S) \in \inparen{\tfrac{1}{10}, \tfrac{1}{2}}$,
$G(S,S) \leq \Gamma_{1-\e/2}(1/10) + o(\sqrt{\e}) \leq \vol(S)(1 - c\sqrt{\e})+o(\sqrt{\e})$ for some $c > 0$.
Thus $\cond_G(S) \geq c'\sqrt{\e}$ for some $c' > 0$.
\end{proof}

The following corollary uses the fact that in the \yes case of \pref{thm:general-sse}, we
actually partition the graph into many non-expanding sets instead of finding just one such set.

\begin{corollary}[Hardness of \mla] \label{cor:mla:app} 
It is \ssehard to approximate \mla to any fixed constant factor.  Formally,
there exists $c>0$ such that for every $\e > 0$, it is 
\ssehard to distinguish between the
following two cases for a given graph $G=(V,E)$, with $|V| = n$:
\begin{description}
\item[\yes:] There exists an ordering $\pi \from V \to [n]$ of the vertices 
such that $\Ex[(u,v) \sim E]{\abs{\pi(u) - \pi(v)}} \leq \e n$
\item[\no:] For all orderings $\pi \from V \to [n]$,  
$\Ex[(u,v) \sim E]{\abs{\pi(u) - \pi(v)}} \geq c \sqrt{\e} n$
\end{description}
\end{corollary}

\begin{proof}
Apply \pref{thm:general-sse} taking 
$q = \lceil 2/\e \rceil, \e' = \e/3$ and $\gamma = \e$.
In the {\yes} case, we pick an arbitrary ordering $\pi$ which orders elements in each of 
the sets $S_1,\ldots,S_q$ contiguously. For these sets, all edges in the set
have length at most $n/q$ and at most $\e' + o(\e)$ fraction of the edges leave the sets.
Thus, 
\[\Ex[(u,v) \sim E]{\abs{\pi(u) - \pi(v)}} ~\leq~ \tfrac{n}{q} + \e' n + o(\e n) ~\leq~ \e n\] 

The proof for the {\no} case follows from an observation of \cite{DevanurKSV06}, 
that for a graph $G$ if every set $S$ with $\vol(S) \in (\tfrac13,\tfrac12)$ has 
$G(S,V\setminus S) \geq \theta$, then for any ordering $\pi \from V \to [n]$,
$\Ex[(u,v) \sim E]{\abs{\pi(u) - \pi(v)}} \geq \tfrac{\theta}{3} \cdot n$ (else one can
obtain a contradiction by optimally ordering the points and cutting randomly between the
positions $n/3$ and $2n/3$). Here, $G(S,V\setminus S) \geq 1/3 - \Gamma_{1-\e/6}(1/3) \geq c'\sqrt{\e}$
for some $c' > 0$. Thus, $\Ex[(u,v) \sim E]{\abs{\pi(u) - \pi(v)}} \geq \tfrac{c'}{3} \cdot n$
for all $\pi \from V \to [n]$.
\end{proof}

\end{document}